\documentclass[final]{beatcs}

\usepackage[utf8x]{inputenc}
\usepackage{fixltx2e}
\usepackage{fix-cm}
\usepackage{ellipsis}
\usepackage{microtype}
\usepackage{graphicx}
\usepackage{booktabs}
\usepackage[table]{xcolor}
\usepackage[english]{babel}
\usepackage[pdfborder=0 0 0,plainpages=true]{hyperref}

\newtheoremstyle{thms}     
  {\topsep}                
  {\topsep}                
  {\itshape}               
  {}                       
  {\bfseries}              
  {}                       
  {5pt plus 1pt minus 1pt} 
  {}                       

\newtheoremstyle{defs}     
  {\topsep}                
  {\topsep}                
  {\upshape}               
  {}                       
  {\bfseries}              
  {}                       
  {5pt plus 1pt minus 1pt} 
  {}                       

\theoremstyle{thms}
\newtheorem{theorem}{Theorem}[section]
\newtheorem{lemma}[theorem]{Lemma}

\newtheorem{proposition}[theorem]{Proposition}
\theoremstyle{defs}
\newtheorem{definition}[theorem]{Definition}
\newtheorem{example}[theorem]{Example}
\newtheorem{remark}[theorem]{Remark}

\def\allFormulae{\mathcal{L}}
\def\allAutoepistemicFormulae{\mathcal{L}_{\mathrm{ae}}}
\def\ie{\emph{i.e.}}
\def\eg{\emph{e.g.}}

\def\true{1}
\def\xor{\oplus}
\def\false{0}
\def\id{\mathrm{id}}

\def\imp{\mathrel{\rightarrow}}
\def\nimp{\mathrel{\nrightarrow}}
\def\leqlogm{\leq^{\log}_m}
\def\equivlogm{\equiv^{\log}_m}
\def\N{\mathbb{N}}
\newcommand\dual[1]{\mathrm{dual}(#1)}
\newcommand\theorems[1]{\mathrm{Th}(#1)}
\newcommand\Vars[1]{\mathrm{Vars}(#1)}
\newcommand\Lits[1]{\mathrm{Lits}(#1)}
\newcommand\SFL[1]{\mathrm{SF}^L(#1)}

\def\CloneBF{\mathsf{BF}}
\def\CloneR{\mathsf{R}}
\def\CloneM{\mathsf{M}}
\def\CloneS{\mathsf{S}}
\def\CloneV{\mathsf{V}}
\def\CloneD{\mathsf{D}}
\def\CloneE{\mathsf{E}}
\def\CloneL{\mathsf{L}}
\def\CloneN{\mathsf{N}}
\def\CloneI{\mathsf{I}}

\newcommand\SigmaP[1]{\Sigma^\mathrm{p}_{#1}}
\newcommand\PiP[1]{\Pi^\mathrm{p}_{#1}}
\newcommand\DeltaP[1]{\Delta^\mathrm{p}_{#1}}
\renewcommand\P{\mathrm{P}}
\newcommand\NP{\mathrm{NP}}
\newcommand\NL{\mathrm{NL}}
\renewcommand\L{\mathrm{L}}
\newcommand\ParityL{\oplus\L}
\newcommand\co{\mathrm{co}}
\newcommand\FP{\mathrm{FP}}
\newcommand\SharpP{\#\P}
\newcommand\SharpCoNP{\#\mathord{\cdot}\co\NP}

\title{Complexity of Non-Monotonic Logics}
\author{Michael Thomas and Heribert Vollmer\thanks{
  Leibniz Universit\"at Hannover,
  Institut f\"ur Theoretische Informatik, 
  \{thomas,vollmer\}@thi.uni-hannover.de. 
  Work partially supported by DFG grant VO 630/6-2.}}

\begin{document}

\maketitle

\section{Introduction}

Over the past few decades, non-monotonic reasoning has developed to be one of the 
most important topics in computational logic and artificial intelligence. 
The non-monotonicity here refers to the fact that, while in usual (monotonic) 
reasoning adding more axioms leads to potentially more possible conclusions, 
in non-monotonic reasoning adding new facts to a knowledge base may prevent 
previously valid conclusions. 
Different ways to introduce non-monotonic aspects to classical logic have been considered:
\begin{enumerate}\def\labelenumi{(\theenumi)} \itemsep 0pt
\item The derivation process may be extended using non-monotonic inference rules.
\item The logical language may be extended with non-monotonic belief operators.
\item The definition of semantics may be changed.
\end{enumerate}

In this survey we consider a logical formalism from each of the above possibilities, namely
\begin{itemize} \itemsep 0pt
\item Reiter's \emph{default logic} (as a candidate for possibility (1) above), which introduces default inference rules of the form $\frac{\alpha:\beta}{\gamma}$, where $\frac{\alpha:\beta}{\gamma}$ intuitively expresses that $\gamma$ can be derived from $\alpha$ as long as $\beta$ is consistent with our knowledge;
\item Moore's \emph{autoepistemic logic} (a candidate for (2)), that extends classical logic with a modal operator $L$ to express the beliefs of an ideal rational agent, in the sense that $L\varphi$ expresses that $\varphi$ is provable;
\item McCarthy's \emph{circumscription} (as candidate for (3)), which restricts the semantics to the minimal models of a formula or set of formulae. 
\end{itemize}
Additionally we survey \emph{abduction}, where one is not interested in inferences from a given knowledge base 
but in computing possible explanations for an observation with respect to a given knowledge base.

Complexity results for different reasoning tasks for propositional variants of these logics have been studied already in the nineties. It was shown that in each case, the complexity is higher than for usual propositional logic (typically complete for some level of the polynomial-time hierarchy). In recent years, however, a renewed interest in complexity issues can be observed. 
One current focal approach is to consider parameterized problems and identify reasonable parameters that allow for FPT algorithms.
In another approach, the emphasis lies on identifying fragments, \ie, restriction of the logical language, that allow more efficient algorithms for the most important reasoning tasks. 

In this survey we focus on this second aspect. We describe complexity results for fragments of logical languages obtained by either restricting the allowed set of operators (\eg, forbidding negations one might consider only monotone formulae) or by considering only formulae in conjunctive normal form but with generalized clause types (which are also called Boolean constraint satisfaction problems). 

The algorithmic problems we consider are suitable variants of satisfiability and implication in each of the logics, but also certain \emph{counting problems}, where one is not only interested in the existence of certain objects (\eg, models of a formula) 
but asks for their number.

\section{Post's Lattice}

  In 1941, Post showed that the sets of Boolean functions closed under projections and arbitrary composition, called \emph{clones}, form a lattice containing only countably infinite such closed sets, and he identified a finite base for each of them~\cite{pos41}. 
  The closure operation, denoted by $[\cdot]$, is not arbitrarily chosen but rather captures an intuitive understanding of expressiveness: Given a set $B$ of Boolean functions, $[B]$ denotes the set of Boolean functions expressible using functions from $B$, or equivalently: computable with Boolean circuits with gates performing functions from $B$. 
  Moreover, it is well behaved with respect to computational complexity;
  \eg , if $\Pi(B)$ is a decision problem defined over Boolean circuits with gates corresponding to Boolean functions from $B$, then $\Pi(B) \leqlogm \Pi(B')$ for all finite sets $B'$ of Boolean functions such that all functions from $B$ can be expressed in $B'$, \ie, $B \subseteq [B']$. 
  Similar statements hold for decision problems defined over Boolean formulae (see~\cite{tho10}).
  \emph{Post's lattice} thus holds the key to study and classify the computational complexity of problems parameterized by finite sets of available Boolean functions.
  In this section, we will define the required terms and notation to introduce Post's lattice. 

  Let $\allFormulae$ be the set of propositional formulae, \ie, 
  the set of formulae defined via
  \[
    \varphi ::= a \mid f(\varphi,\ldots,\varphi), 
  \]
  where $a$ is a proposition and $f$ is an $n$-ary Boolean function (we do not distinguish between connectives and their associated functions).
  For a finite set $B$ set of Boolean functions, a \emph{$B$-formula} is a Boolean formula using functions from $B$ only. The set of all $B$-formulae is denoted by $\allFormulae(B)$.

  A \emph{clone} is a set of Boolean functions that is closed under superposition, \emph{i.e.},
  $B$ contains all projections (that is, the functions $\mathrm{I}^n_m(x_1,\ldots,x_n) = x_m$ for $n  \in \N$ and $1 \leq m \leq n$) 
  and is closed under arbitrary composition \cite{pip79}. 
  For a set $B$ of Boolean functions, we denote by $[B]$ the smallest clone containing $B$ and 
  call $B$ a \emph{base} for $[B]$. 
  A $B$-formula $g$ is called \emph{$B$-representation} of $f$ if $f \equiv g$.

  Post showed that the set of all clones ordered by inclusion together with meet $A\wedge B=[A \cap B]$ and join $A\vee B=[A \cup B]$ forms the lattice depicted in Figure~\ref{fig:post's_lattice}.
  To give the list of all the clones, we need the following properties. 
  Say that a set $A \subseteq \{\false,\true\}^n$ is \emph{$c$-separating}, $c \in \{\false,\true\}$, 
  if there exists an $i \in \{1, \ldots , n\}$ such that $(a_1,\ldots,a_n) \in A$ implies $a_i = c$. 
  Let $f$ be an $n$-ary Boolean function and define the dual of $f$ to be the Boolean function 
  $\dual{f}(x_1, \ldots , x_n) := \neg f(\neg x_1, \ldots , \neg x_n)$.
  We say that:
  \begin{itemize} \itemsep 0pt
    \item $f$ is \emph{$c$-reproducing} if $f(c, \ldots , c) = c$, $c \in \{\false,\true\}$;
    \item $f$ is \emph{$c$-separating} if $f^{-1}(c)$ is $c$-separating, $c \in \{\false,\true\}$;
    \item $f$ is \emph{$c$-separating of degree $m$} if all $A \subseteq f^{-1}(c)$ with $|A|=m$ are $c$-separating;
    \item $f$ is \emph{monotone} if $a_1 \leq b_1, \ldots, a_n \leq b_n$ implies $f(a_1, \ldots , a_n) \leq f(b_1, \ldots , b_n)$;
    \item $f$ is \emph{self-dual} if $f \equiv \dual{f}$; here, $\dual{f}(x_{1},\dots,x_{n})=\neg f(\neg x_{1},\dots,\neg x_{n})$;
    \item $f$ is \emph{affine} if $f(x_1,\ldots,x_n) \equiv x_1 \xor \cdots \xor x_n \xor c$ with $c \in \{\false,\true\}$;
    \item $f$ is \emph{essentially unary} if $f$ depends on at most one variable.
  \end{itemize}
  The above properties canonically extend to sets of Boolean functions. 
  The list of all clones is given in Table~\ref{tab:clones}, 
  where $\id$ denotes the identity $\mathrm{I}^{1}_{1}$ and 
  $\mathrm{T}^{n+1}_n$ is the $(n+1)$-ary threshold function with threshold $n$ defined as $\mathrm{T}^{n+1}_n(x_0,\ldots,x_{n}) := \bigvee_{i=0}^{n} (x_0\land \cdots \land x_{i-1} \land x_{i+1} \land \cdots \land x_n)$, \emph{i.e.}, $\mathrm{T}^{n+1}_n(x_0,\ldots,x_n)$ is $1$ if at least $n$ of its inputs are $1$. 

  \begin{table}
    \fontsize{9.5}{12}\selectfont
    \setlength{\aboverulesep}{0pt}
    \setlength{\belowrulesep}{0pt}
    \rowcolors{2}{gray!10}{}
    \centering
    \begin{tabular}{lll}
    \toprule
    Clone & Definition & Base \\
    \midrule
    $\CloneBF$    & All Boolean functions & $\{x \land y ,\neg x\}$ \\
    $\CloneR_0$   & $\{f \in \CloneBF \mid f \text{ is $\false$-reproducing}\}$ & $\{ x \land y , x \xor y \}$ \\
    $\CloneR_1$   & $\{f \in \CloneBF \mid f \text{ is $\true$-reproducing}\}$ & $\{x \lor y, x \leftrightarrow y \}$ \\
    $\CloneR_2$   & $\CloneR_0 \cap \CloneR_1$ & $\{x \lor y, x \land (y \leftrightarrow z) \}$ \\
    $\CloneM$     & $\{f \in \CloneBF \mid f \text{ is monotone}\}$ & $\{x \land y, x\lor y,\false,\true\}$ \\
    $\CloneM_0$   & $\CloneM \cap \CloneR_0$ & $\{x \land y,x\lor y,\false\}$ \\
    $\CloneM_1$   & $\CloneM \cap \CloneR_1$ & $\{x \land y,x\lor y,\true\}$ \\  
    $\CloneM_2$   & $\CloneM \cap \CloneR_2$ & $\{x \land y,x\lor y\}$ \\
    $\CloneS_0$   & $\{f \in \CloneBF \mid f \text{ is $\false$-separating}\}$ & $\{x \imp y\}$ \\
    $\CloneS_0^n$ & $\{f \in \CloneBF \mid f \text{ is $\false$-separating of degree $n$}\}$ & $\{x \imp y,\dual{\mathrm{T}^{n+1}_n}\}$ \\
    $\CloneS_1$   & $\{f \in \CloneBF \mid f \text{ is $\true$-separating}\}$ & $\{x \nimp y\}$ \\
    $\CloneS_1^n$ & $\{f \in \CloneBF \mid f \text{ is $\true$-separating of degree $n$}\}$ & $\{x \nimp y,\mathrm{T}^{n+1}_n\}$ \\
    $\CloneS_{02}^n$& $\CloneS_0^n \cap \CloneR_2$ & $\{x \lor (y \land \neg z), \dual{\mathrm{T}^{n+1}_n}\}$ \\
    $\CloneS_{02}$  & $\CloneS_0 \cap \CloneR_2$ & $\{x \lor (y \land \neg z)\}$ \\
    $\CloneS_{01}^n$& $\CloneS_0^n \cap \CloneM$ & $\{\dual{\mathrm{T}^{n+1}_n},\true\}$ \\
    $\CloneS_{01}$  & $\CloneS_0 \cap \CloneM$ & $\{x \lor (y \land z),\true\}$ \\
    $\CloneS_{00}^n$& $\CloneS_0^n \cap \CloneR_2 \cap \CloneM$ & $\{x \lor (y \land z),\dual{\mathrm{T}^{n+1}_n}\} $ \\
    $\CloneS_{00}$  & $\CloneS_0 \cap \CloneR_2 \cap \CloneM$ & $\{x \lor (y \land z)\}$ \\
    $\CloneS_{12}^n$& $\CloneS_1^n \cap \CloneR_2$ & $\{x\land (y\lor \neg z),\mathrm{T}^{n+1}_n\}$ \\
    $\CloneS_{12}$  & $\CloneS_1 \cap \CloneR_2$ & $\{x\land (y\lor \neg z)\}$ \\
    $\CloneS_{11}^n$& $\CloneS_1^n \cap \CloneM$ & $\{\mathrm{T}^{n+1}_n,\false\}$ \\
    $\CloneS_{11}$  & $\CloneS_1 \cap \CloneM$ & $\{x\land (y\lor z),\false\}$ \\
    $\CloneS_{10}^n$& $\CloneS_1^n \cap \CloneR_2 \cap \CloneM$ & $\{x\land (y\lor z),\mathrm{T}^{n+1}_n\}$ \\
    $\CloneS_{10}$  & $\CloneS_1 \cap \CloneR_2 \cap \CloneM$ & $\{x\land (y\lor z)\}$ \\
    
    $\CloneD$     & $\{f \in \CloneBF \mid f \text{ is self-dual}\}$ & $\{(x \land y) \lor (x \land \neg z) \lor (\neg y \land \neg z)\}$ \\
    $\CloneD_1$   & $\CloneD \cap \CloneR_2$ & $\{(x \land y) \lor (x \land \neg z) \lor (y \land \neg z)\}$ \\
    $\CloneD_2$   & $\CloneD \cap \CloneM$ & $\{(x \land y) \lor (x \land z) \lor (y \land z)\}$ \\  
    
    $\CloneL$     & $\{f \in \CloneBF \mid f \text{ is affine}\}$ & $\{x \xor y,\true\}$ \\
    $\CloneL_0$   & $\CloneL \cap \CloneR_0$ & $\{x \xor y\}$ \\
    $\CloneL_1$   & $\CloneL \cap \CloneR_1$ & $\{x \leftrightarrow y\}$ \\
    $\CloneL_2$   & $\CloneL \cap \CloneR_2$ & $\{x \xor y \xor z\}$ \\
    $\CloneL_3$   & $\CloneL \cap \CloneD$ & $\{x \xor y \xor z \xor \true \}$ \\
    
    $\CloneE$     & $\{f \in \CloneBF \mid f \text{ is constant or a conjunction}\}$ & $\{x \land y,\false, \true\}$ \\
    $\CloneE_0$   & $\CloneE \cap \CloneR_0$ & $\{x \land y,\false\}$ \\
    $\CloneE_1$   & $\CloneE \cap \CloneR_1$ & $\{x \land y,\true\}$ \\
    $\CloneE_2$   & $\CloneE \cap \CloneR_2$ & $\{x \land y\}$ \\
    
    $\CloneV$     & $\{f \in \CloneBF \mid f \text{ is constant or a disjunction}\}$ & $\{x \lor y,\false,\true\}$ \\
    $\CloneV_0$   & $\CloneV \cap \CloneR_0$ & $\{x \lor y,\false\}$ \\
    $\CloneV_1$   & $\CloneV \cap \CloneR_1$ & $\{x \lor y,\true\}$ \\
    $\CloneV_2$   & $\CloneV \cap \CloneR_2$ & $\{x \lor y\}$ \\
    
    $\CloneN$     & $\{f \in \CloneBF \mid f \text{ is essentially unary}\}$ & $\{\neg x,\false,\true\}$ \\ 
    $\CloneN_2$   & $\CloneN \cap \CloneD$ & $\{\neg x\}$ \\

    $\CloneI$     & $\{f \in \CloneBF \mid f \text{ is constant or a projection}\}$ & $\{\id,\false,\true\}$ \\ 
    $\CloneI_0$   & $\CloneI \cap \CloneR_0$ & $\{\id,\false\}$ \\
    $\CloneI_1$   & $\CloneI \cap \CloneR_1$ & $\{\id,\true\}$ \\
    $\CloneI_2$   & $\CloneI \cap \CloneR_2$ & $\{\id\}$ \\
    \bottomrule
    \end{tabular}
    \caption{List of all clones with definition and bases}
    \label{tab:clones}
  \end{table}

  \begin{figure}
  \centering
  \includegraphics[width=\textwidth]{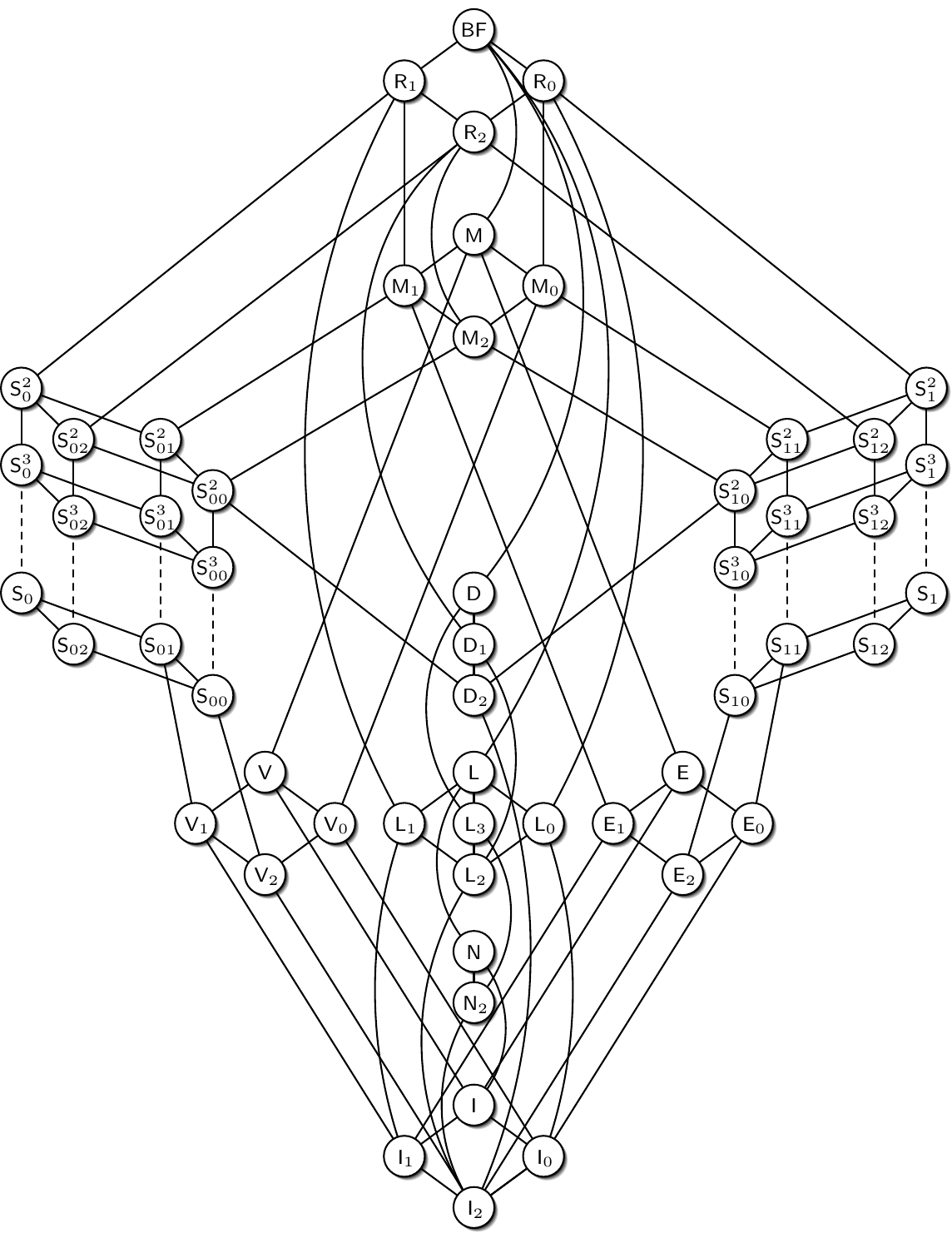}
  \caption{\label{fig:post's_lattice}Post's lattice}
  \end{figure}

  The following easy observation, which will be useful in the subsequent sections, gives a 
  a first example of the relationship between function-restricted sets of Boolean formulae.

  \begin{lemma} \label{lem:discard_true}
    Let $B$ be a finite set of Boolean functions and let $L$ be a set of $(B \cup \{\true\})$-formulae.
    Then $L$ can be transformed in logspace into a set $L'$ of $B$-formulae such that the number of satisfying assignments of $L$ and $L'$ coincide.
  \end{lemma}

  \begin{proof}[Proof sketch]
    The idea is to add a fresh proposition $t$ to $L$ and to replace all occurrences of the constant $\true$ with $t$. 
    As a result we obtain that, for problems $\Pi$ defined over sets of Boolean formulae, $\Pi(B \cup \{\true\}) \equivlogm \Pi(B)$.
  \end{proof}

\section{Default Logic}

  Default logic is among the best known and most successful formalisms for non-monotonic reasoning.
  It was proposed by Raymond Reiter in 1980~\cite{reiter:80} and extends classical logic with 
  \emph{default rules}, \ie, defeasible inference rules with an additional justification.
  These capture the process of deriving conclusions based on inferences of the form 
  ``in the absence of contrary information, assume \ldots''.
  As with few exceptions most of our knowledge about the world 
  is almost true rather than an absolute truth,
  Reiter argued that his logic is an adequate formalization of the human reasoning under the
  \emph{closed world assumption}, which allows one to assume the negation of all facts not 
  derivable from the given knowledge base.

  Formally, a default rule is an expression of the form $\frac{\alpha:\beta}{\gamma}$, 
  where $\alpha$, $\beta$, $\gamma$ are formulae; $\alpha$ is called the \emph{premise}, 
  $\beta$ is called \emph{justification}, and $\gamma$ is called \emph{conclusion}. 
  Further, a \emph{default theory} is a pair $(W,D)$, where $W$ is a set of formulae and $D$ is a set of default rules.
  The intended interpretation of default rules is that $\gamma$ holds 
  if $\alpha$ can be derived and $\beta$ is consistent with our knowledge and beliefs about the world. 
  It is this consistency condition that introduces the non-monotonicity: 
  \begin{example}
    Consider the default theory $(W,D)$ with 
    $
      W:=\{x\}, D:=\left\{\frac{x:y}{z}\right\}.
    $
    Then we should be able to derive $z$ from $(W,D)$, as $x$ is a fact from $W$ and $y$ is consistent with the consequences of $W \cup \{z\}$. However, if $W$ is extended with $\{\neg y\}$ then $W$ is no longer consistent with the justification of $\frac{x:y}{z}$ and we have no right to conclude $z$. The addition of $\neg y$ to $W$ thus invalidates the consequence $z$.  
  \end{example}

  As another consequence, the derivable knowledge depends on the set of applied defaults:
  \begin{example}\label{ex:dl-2}
    Consider the default theory $(\emptyset,D)$ with 
    $
      D:=\left\{\frac{\true:x}{\neg y}, \frac{\true:y}{\neg x}\right\}.
    $
    If we apply the left default, then $\neg y$ is derivable while the right default is blocked, as its justifications is inconsistent with the conclusion $\neg y$. On the other hand, if we apply the right default first, then  $\neg x$ is derived and the left default rule gets blocked.
  \end{example}

  Thus, to appropriately represent the knowledge derivable from a default theory, we introduce the notion of stable extensions.

  \begin{definition}[{\cite{reiter:80}}] \label{def:stable-extension}
    Let $(W,D)$ be a default theory and $E$ be a set of formulae.
    Let $E_0:=W$ and
    $
      E_{i+1} := \theorems{E_i} \cup \left\{\gamma \,\middle|\, \frac{\alpha : \beta}{\gamma} \in D, \alpha \in E_i \text{ and } \neg \beta \notin E \right\}.
    $
    Then $E$ is a stable extension of $(W,D)$ if and only if $E= \bigcup_{i \in \N} E_i$.
  \end{definition}

  Stable extensions can alternatively be characterized as the least fixed points of an operator $\Gamma_{W,D}$:
  For a given default theory $(W,D)$ and a set $E$ of formulae, 
  let $\Gamma_{W,D}(E)$ be the smallest set of formulae such that 
  \begin{enumerate} \itemsep 0pt
    \item 
    $W\subseteq \Gamma(E)$,
    \item
    $\Gamma(E)$ is deductively closed (\ie, $\Gamma_{W,D}(E)=\theorems{\Gamma_{W,D}(E)}$), and
    \item
    for all $\frac{\alpha:\beta}{\gamma} \in D$ 
    with $\alpha \in \Gamma_{W,D}(E)$ and $\neg \beta \notin E$, it holds that $\gamma \in \Gamma_{W,D}(E)$\\
    (in this case, we also say that the default $\frac{\alpha:\beta}{\gamma}$ is \emph{applicable}).
  \end{enumerate}

  \begin{proposition}[{\cite{reiter:80}}]
    Let $(W,D)$ be a default theory and a $E$ be a set of formulae. 
    Then $E$ is a stable extension of $(W,D)$ iff $E$ is a fixed point of $\Gamma_{W,D}$.
  \end{proposition}

  We have already observed in Example~\ref{ex:dl-2} that a default theory may possess several stable extensions; 
  indeed, a default theory with $n$ default rules may have any number of different stable extensions between $0$ and $2^n$.
  Thus the following questions naturally arise:
  
  (\emph{Extension existence}) Does a given default theory admit a stable extension? 
  
  (\emph{Credulous reasoning}) Is a given formula contained in at least one stable extension of a given default theory? 
  
  (\emph{Skeptical reasoning}) And, finally, is a given formula contained in all stable extensions of a given default theory?

  The computational complexity of the corresponding decision problems was first explored by Kautz and Selman~\cite{kautz-selman:91}, and Stillman~\cite{stillman:90}, who both presented results for syntactically restricted fragments of disjunction-free default logics.
  In 1992, Gottlob~\cite{gottlob:92} and Stillman~\cite{stillman:92} then independently showed that the computational complexity of these questions is presumably higher than that of the corresponding satisfiability and implication problem in propositional logic:

  \begin{theorem}[\cite{gottlob:92,stillman:92}] \label{thm:default-logic}
    The extension existence problem and the credulous reasoning problem for default logic are $\SigmaP2$-complete, 
    whereas the skeptical reasoning problem for default logic is $\PiP2$-complete.
  \end{theorem}

  More recently, 
    Liberatore and Schaerf~\cite{LiberatoreS05} showed that model checking (\ie, the task to decide whether a given assignment is a model of any extension of a given default theory) is $\SigmaP2$-complete, too. 
  And 
    Ben-Eliyahu-Zohary~\cite{Ben-Eliyahu-Zohary02} extended the complexity landscape of default logics with results on disjunction-free fragments dual to those studied by Kautz and Selman, and Stillman.
  The study of these fragments was motivated by embeddings of other formalisms into default logic. 
  However, little was known about the complexity of not-disjunction-free default logics.
  In \cite{bemethvo08}, the authors devise a systematic study of the fragments of default logic obtained by restricting the set of available Boolean functions. 
  The results provide insight into the source of the hardness of default reasoning and 
  reveal the trade-off between expressivity and computational complexity of fragments of default logic.

  To present the results, let $B$ be a finite set of Boolean functions.
  Say that the default theory $(W,D)$ is a \emph{$B$-default theory} if
  $W \cup \left\{\alpha,\beta,\gamma\,\middle|\,\frac{\alpha:\beta}{\gamma}\right\} \subseteq \allFormulae(B)$ and 
  let $B$-default logic denote default logic restricted to $B$-default theories.

  \begin{theorem}[{\cite{bemethvo08}}] \label{thm:extension-existence}
    Let $B$ be a finite set of Boolean functions. 
    Then the extension existence problem for $B$-default logic is 
    \begin{enumerate} \itemsep 0pt
      \item
      \label{thm:extension-existence-SigmaP2} 
      $\SigmaP2$-complete if $\CloneS_{1} \subseteq [B]$ or $\CloneD \subseteq [B]$,
      \item 
      \label{thm:extension-existence-DeltaP2} 
      $\DeltaP2$-complete if $\CloneS_{11} \subseteq [B] \subseteq \CloneM$,
      \item 
      \label{thm:extension-existence-NP} 
      $\NP$-complete if $[B] \in \{\CloneN,\CloneN_2,\CloneL,\CloneL_0,\CloneL_3\}$,
      \item 
      \label{thm:extension-existence-AC0} 
      $\P$-complete if $[B] \in \{\CloneV,\CloneV_0,\CloneE,\CloneE_0\}$,
      \item 
      \label{thm:extension-existence-NL} 
      $\NL$-complete if $[B] \in \{\CloneI,\CloneI_0\}$, and
      \item 
      \label{thm:extension-existence-trivial} 
      trivial in all other cases (that is, if $[B] \subseteq \CloneR_1$),
    \end{enumerate}
    via logspace many-one reductions.
  \end{theorem}
  
  A key observation to the proof of this theorem is the following lemma.
  
  \begin{lemma} \label{lem:dl-only-one-ext}
    Let $B$ be a finite set of Boolean functions.
    If $[B] \subseteq \CloneM$ then any $B$-default theory has at most one stable extension;
    if $[B] \subseteq \CloneR_1$ then any $B$-default theory has at exactly one stable extension.
  \end{lemma}
  
  \begin{proof}
    Suppose first that $[B] \subseteq \CloneM$. 
    Then each function $f \in B$ is either $\true$-reproducing or equivalent to $\false$.
    Default rules with a justification equivalent to $\false$ are not applicable unless $W$ is inconsistent.
    As in this case $\allFormulae$ is the only stable extension of the default theory~\cite[Corollary~2.3]{reiter:80}, 
    we w.l.o.g.\ suppose that all justifications are $\true$-reproducing, \ie, $[B] \subseteq \CloneR_1$.
    Now observe that the negation of a $\true$-reproducing function is not $\true$-reproducing 
    while all consequents of $\true$-reproducing functions are. 
    Indeed, if $[B] \subseteq \CloneR_1$ then any stable extension of a default theory is consistent and satisfied 
    by the assignment setting to $\true$ all propositions, which also satisfies the justification of each default rule.
    Thus any monotone default theory may possess at most one stable extension
    while any $\true$-reproducing default theory possesses exactly one stable extension.
  \end{proof}
  
  We will sketch the proof of Theorem~\ref{thm:extension-existence}.
  
  \begin{proof}[Proof sketch]
    For $\CloneS_{1} \subseteq [B]$ or $\CloneD \subseteq [B]$, the $\SigmaP2$-hardness follows from Theorem~\ref{thm:default-logic} and Lemma~\ref{lem:discard_true}, as $[\CloneS_1 \cup \{\true\}]=[\CloneD \cup \{\true\}]=\CloneBF$ and the upper bound easily generalizes from $\{\land,\lor,\neg\}$ to arbitrary sets of Boolean functions.
    
    The case $[B] \subseteq \CloneR_1$ follows directly from Lemma~\ref{lem:dl-only-one-ext}.
    It hence remains to consider those sets $B$ such that $[B \cup \{\true\}]$ contains the constant $\false$.
    These are all included in either the clone $\CloneM$ or the clone $\CloneL$ (or both).

    For $\CloneS_{11} \subseteq [B] \subseteq \CloneM$, membership in $\DeltaP2$ follows similarly from Lemmas~\ref{lem:discard_true} and~\ref{lem:dl-only-one-ext}:
    the only way for a monotone default theory not to possess a stable extension is to contain a default rule $\frac{\alpha:\beta}{\gamma}$ such that $\gamma \equiv \false$. It thus suffices to compute the set of applicable defaults using subsequent calls to a $\co\NP$-oracle for $B$-formula implication and to verify that their conclusions are satisfied by the assignment setting to $\true$ all propositions.
    It is straightforward to implement this as a $\DeltaP2$-algorithm.
    The $\DeltaP2$-hardness on the other hand is established using a reduction from the \emph{sequentially nested satisfiability problem}, which was first identified to be $\DeltaP2$-complete in~\cite[Theorem~3.4]{got95b} (see also \cite{laroussinie-markey-schnoebelen:01}). 
    
    If one further restricts the set $B$ such that $[B] \in \{\CloneV,\CloneV_0,\CloneE,\CloneE_0\}$ (\ie, $[B \cup \{\true\}]$ does contain the Boolean constants and either conjunctions or disjunctions), then formula implication and the the extension existence problem become tractable~\cite{beyersdorff-meier-thomas-vollmer:09a}. Indeed, the problem becomes $\P$-complete, as can be shown by a reduction from
    a variant of the circuit value problem.
    
    Further restricting the set $B$ such that $[B] \in \{\CloneI,\CloneI_0\}$ (\ie, $[B \cup \{\true\}]$ does only contain the Boolean constants) leads to default theories with rules whose premise and conclusion are a single proposition. 
    As a result, the existence of a stable extension reduces to the complement of the reachability problem in directed graphs. Similarly, reachability in such graphs can be transformed into the question whether a stable extension does not have a stable extension. As this problem is $\NL$-complete and $\NL$ is closed under complement, the extension existence problem for such default theories is $\NL$-complete.
    
    Finally, for $[B] \in \{\CloneL,\CloneL_0,\CloneL_3,\CloneN,\CloneN_2\}$ (\ie, $\neg \in [B \cup \{\true\}]$ and $[B]$ is affine), we observe a different situation. There may exist exponentially many different stable extensions; yet, the verification of a candidate is tractable because implication and satisfiability of $B$-formulae are~\cite{beyersdorff-meier-thomas-vollmer:09a}. Hence, the extension existence problem becomes solvable in $\NP$. The $\NP$-hardness, on the other hand, is obtained by reducing from the satisfiability problem for 3CNF formulae: given a formula $\varphi \equiv \bigwedge_{i=1}^n c_i$ with $c_i \equiv \ell_{i1} \lor \ell_{i2} \lor \ell_{i3}$, we construct the default theory $(\emptyset,D)$ with
    \[
      D:= \bigg\{ \frac{\true : x_i}{x_i},\frac{\true : \neg x_i}{\neg x_i} \,\bigg|\, x_i \in \Vars{\varphi} \bigg\} \cup 
          \bigg\{ \frac{\overline{\ell_{i1}} : \overline{\ell_{i2}}}{\ell_{i3}} \,\bigg|\, 1 \leq i \leq n \bigg\},
    \]
    where, for a literal $\ell$, $\overline{\ell}$ denotes the literal of opposite polarity, and for a formula $\varphi$, $\Vars{\varphi}$ denotes the set of all variables in $\varphi$. 
    It is easy to verify the correctness of this reduction.
    As the above default theory can easily be written as a $B$-default theory for all $B$ such that $\neg \in [B]$, the proof is complete.
  \end{proof}

  \begin{remark}
    As default rules require the justification $\beta$ to be consistent with a stable extension $E$ (\ie, $\neg \beta \notin E$),
    another conceivable formalization of $B$-default logic would be to require $\alpha$ and $\gamma$ to be $B$-formulae and 
    $\beta$ to be the negation of a $B$-formula. For this formalization, the extension existence problem for $B$-default logic is 
      $\SigmaP2$-complete if $\CloneS_{00} \subseteq [B]$ or $\CloneS_{10} \subseteq [B]$ or $\CloneD_{2} \subseteq [B]$, and 
      tractable otherwise (with this case splitting into $\ParityL$-complete cases and logspace-solvable cases).
  \end{remark}

  Given the upper and lower bounds for the stable extension problem it is easy to settle the complexity of the credulous and skeptical reasoning problem. Define the \emph{credulous (resp.\ skeptical) reasoning problem for $B$-default logic} as the problem to decide, given a $B$-default theory $(W,D)$ and a $B$-formula $\varphi$, whether $\varphi$ is contained in a stable extension (resp.\ all stable extensions)  of $(W,D)$.
  
  \begin{theorem}[{\cite{bemethvo08}}] \label{thm:credulous-reasoning-dl}
    Let $B$ be a finite set of Boolean functions. 
    Then the credulous (resp.\ skeptical) reasoning problem for $B$-default logic is 
    \begin{enumerate} \itemsep 0pt
      \item
      $\SigmaP2$-complete (resp.\ $\PiP2$-complete) if $\CloneS_{1} \subseteq [B]$ or $\CloneD \subseteq [B]$,
      \item
      $\DeltaP2$-complete if $\CloneS_{11} \subseteq [B] \subseteq \CloneM$, 
      \item
      $\co\NP$-complete if $\CloneS_{00} \subseteq [B] \subseteq \CloneR_1$ or $\CloneS_{10} \subseteq [B] \subseteq \CloneR_1$ or $\CloneD_{2} \subseteq [B] \subseteq \CloneR_1$, 
      \item
      $\NP$-complete (resp.\ $\co\NP$-complete) if $[B] \in \{\CloneN,\CloneN_2,\CloneL,\CloneL_0,\CloneL_3\}$,
      \item
      $\P$-complete if $\CloneV_2 \subseteq [B] \subseteq \CloneV$ or $\CloneE_2 \subseteq  [B] \subseteq \CloneE$ or $[B]\in\{\CloneL_1,\CloneL_2\}$,
      and
      \item
      $\NL$-complete in all other cases (that is, if $[B] \subseteq \CloneI$),
    \end{enumerate}
    via logspace many-one reductions.
  \end{theorem}
  
  Another possibility to study fragments of default logic, aside restricting the available Boolean functions, is \emph{Schaefer's framework}~\cite{sch78}. This framework is motivated by constraint satisfaction problem, where a set of conditions represented as logical relations has to be simultaneously satisfied. 
  Hence, the set $W$ and the formulae occurring in $D$ are assumed to be a set of applications of relations from $S$ to the variables in $\Vars{W} \cup \Vars{D}$, where $\Vars{D}$ is a shorthand for $\mathrm{Vars}\Big(\big\{\alpha,\beta,\gamma \,\big|\, \frac{\alpha:\beta}\gamma \in D\big\}\Big)$. That is, $W$ and the formulae occurring in $D$ are of the form 
  $\{R_1(x_{11},\ldots,x_{1m_1}),\ldots, R_n(x_{n1},\ldots,x_{nm_n})\}$,
  where the $R_i$'s are relations of arity $m_{i}$ from a fixed set $S$ of available relations over the domain $\{\false,\true\}$ and the variables $x_{11},\ldots,x_{nm_n}$ are from $\Vars{W} \cup \Vars{D}$. Such a set of applications of relations is correspondingly said to be satisfied by an assignment $\sigma$ if  $(\sigma(x_{i1}),\ldots,\sigma(x_{im_i})) \in R_i$ for all $1 \leq i \leq n$.
  Call a relation $R$ \emph{Schaefer} if it is either
  \begin{itemize} \itemsep 0pt
    \item \emph{affine} (coincides with the set of models of an $\{\xor\}$-formula),
    \item \emph{bijunctive} (coincides with the set of models of a 2CNF formula),
    \item \emph{Horn} (coincides with the set of models of a Horn formula), or
    \item \emph{dual Horn} (coincides with the set of models of a dual Horn formula).
  \end{itemize}
  And say that a set of relations is Schaefer if there is one of the above four properties that is satisfied by all relations in $S$.
  Call a default theory $(W,D)$ such that 
  $W \cup \left\{\alpha,\beta,\gamma\,\middle|\,\frac{\alpha:\beta}{\gamma}\right\}$ is a set of applications of relations from $S$
  a \emph{default theory over relations from $S$}.
  
  We define the \emph{extension existence problem for default logic over relations from $S$} as the problem to decide,
  given a default theory $(W,D)$ over relations from $S$, whether $(W,D)$ has a stable extension.
  Further, define the \emph{credulous} (resp.\ \emph{skeptical}) \emph{reasoning problem for default logic over relations from $S$} as the problem to decide,
  given a default theory $(W,D)$ over relations from $S$ and a set $\varphi$ of applications of relations from $S$, whether $\varphi$ is contained in at least one (resp.\ any) stable extension of $(W,D)$.
  In \cite{chhesc07}, Chapdelaine \emph{et al.}\  study the complexity of these problems and establish the following trichotomies:
  
  \begin{theorem}[\cite{chhesc07,schnoor07}]
    Let $S$ be a set of relations. Then the 
    extension existence problem for default logic over relations from $S$ is
    \begin{enumerate} \itemsep 0pt
      \item $\SigmaP2$-complete if $S$ is not Schaefer,
      \item $\NP$-complete if $S$ is Schaefer but neither $\false$-valid or $\true$-valid,
      \item in $\P$ in all other cases.
    \end{enumerate}
  \end{theorem}
  
  \begin{theorem}[\cite{chhesc07,schnoor07}]
    Let $S$ be a set of relations. Then the 
    credulous (resp.\ skeptical) reasoning problem for default logic over relations from $S$ is
    \begin{enumerate} \itemsep 0pt
      \item $\SigmaP2$-complete if $S$ is not Schaefer,
      \item $\NP$-complete (resp.\ $\co\NP$-complete) if $S$ is Schaefer but neither $\false$-valid nor $\true$-valid,
      \item $\co\NP$-complete if $S$ $\false$-valid or $\true$-valid but not Schaefer,
      \item in $\P$ in all other cases.
    \end{enumerate}
  \end{theorem}
  For detailed proofs of these results, see \cite{schnoor07}. 
  
  We like to remark that also results like these about Boolean constraint satisfaction problems are proved using Post's lattice. This is because the classes of Boolean relations above (affince, bijunctive, Horn, dual Horn) can all be defined using so called \emph{polymorphism}, a kind of closure property, of Boolean relations. These sets of polymorphism are always clones, \ie, they appear somewhere in the lattice. To state only one example, a relation is Horn iff the set of its polymorphisms is the class $\CloneE_{2}$. The structure of the lattice is then used in the proof to make a case distinction on all possible sets of polymorphisms of $S$ and determine the complexity in each case. For more details, we refer the reader to \cite{crvo08}.

  Having settled the complexity of these decision problems, 
  mind that these results do only speak about the existence of objects, \eg, stable extensions.
  But what about the complexity of counting them? 
  We will conclude this survey of the complexity of default logic with a treatment of the problem to count the number stable extensions.
  
  Let us introduce the relevant notions and counting complexity classes first. 
  For alphabets $\Sigma$ and $\Pi$, 
  let $A \subseteq \Sigma^\star \times \Pi^\star$ be a binary relation such that the set 
  $A(x) := \{y \in \Pi^\star \mid (x,y) \in A\}$ is finite for all $x \in \Sigma^\star$. 
  We write $\#A$ to denote the following \emph{counting problem}:
  Given $x \in \Sigma^\star$, compute $|A(x)|$.
  The class of counting problems computable in polynomial time is denoted by $\FP$.
  To characterize the complexity of counting problems that are not known to be in $\FP$, we follow \cite{hevo95} 
  and define an operator $\#\mathord{\cdot}\mathcal{C}$ on classes $\mathcal{C}$ of decision problems:
  $\#A \in \#\mathord{\cdot}\mathcal{C}$ if (a) there exists a polynomial $p$, such that for all $x$ and all $y \in A(x)$,
  $|y| \leq p(|x|)$ and (b) the problem to decide, given $x$ and $y$, whether $y \in A(x)$ is in $\mathcal{C}$.
  Clearly, $\#\mathord{\cdot}\P$ coincides with $\SharpP$, the class of functions counting the number of accepting path of nondeterministic polynomial-time Turing machines---the natural analogue of $\NP$ in the counting complexity context~\cite{val79b}. Applying $\#\cdot$ to the classes of the polynomial hierarchy, we now obtain a linearly ordered hierarchy of counting complexity classes~\cite{tod91b,hevo95}:
  $
    \SharpP \subseteq \#\mathord{\cdot}\NP \subseteq \#\mathord{\cdot}\co\NP = \#\mathord{\cdot}\P^{\NP} \subseteq \#\mathord{\cdot}\SigmaP2 \subseteq \#\mathord{\cdot}\PiP2 = \#\mathord{\cdot}\P^{\SigmaP2} \subseteq \cdots.
  $
  
  The counting complexity of default logic has, to the authors' best knowledge, first been considered in \cite{thomas:10:phd}.
  There, it was shown that counting the number of stable extensions is 
  complete for the second level of the counting polynomial hierarchy, $\SharpCoNP$, whenever $[B \cup \{\true\}]=\CloneBF$;
  becomes $\DeltaP2$-complete for all sets $B$ such that $[B \cup \{\true\}]=\CloneM$;
  complete for the first level of the counting hierarchy for all affine sets $B$ such that $\neg$ can be implemented from $B \cup \{\true\}$; 
  and becomes efficiently computable in all other cases.
  The counting complexity thus decreases analogously to the complexity of the extension existence problem.
  However, observe that we blur over the distinction between decision problems and their characteristic functions:
  By Lemma~\ref{lem:dl-only-one-ext} any monotone $B$-default theory has at most one stable extension.
  The problem to count the number stable extensions thus coincides with the characteristic function of the extension existence problem, 
  which $\DeltaP2$-complete.
  
  \begin{theorem}[\cite{thomas:10:phd}] \label{thm:num-extensions}
    Let $B$ be a finite set of Boolean functions. 
    Then the problem to count the number of stable extensions in $B$-default logic is 
    \begin{enumerate} \itemsep 0pt
      \item 
      $\SharpCoNP$-complete if $\CloneS_{1} \subseteq [B]$ or $\CloneD \subseteq [B]$,
      \item 
      $\DeltaP2$-complete if $\CloneS_{11} \subseteq [B] \subseteq \CloneM$,
      \item 
      $\SharpP$-complete if $[B] \in \{\CloneN,\CloneN_2,\CloneL,\CloneL_0,\CloneL_3\}$,
      \item 
      in $\FP$ in all other cases (that is, if $[B]  \subseteq \CloneV$ or $[B] \subseteq \CloneE$ or $[B] \subseteq \CloneR_1$)
    \end{enumerate}
    via parsimonious reductions.
  \end{theorem}
  
  Note that for the classification in Theorem~\ref{thm:num-extensions} the conceptually simple 
  and well-behaved parsimonious reductions are sufficient (a counting problem $\#A$ parsimoniously reduces to a counting problem $\#B$ if there is a polynomial-time computable function $f$ such that for all inputs $x$, $|A(x)|=|B(f(x))|$ \cite{val79b}), while for related classifications 
  in the literature less restrictive and more complicated reductions had to be used 
  (see, \eg, \cite{duheko05,hermann-pichler:07, babocrrescvo10}).

\section{Autoepistemic Logic}

  Autoepistemic logic was introduced by Moore~\cite{moore:85} to overcome some of the peculiarities of the 
  non-monotonic logics devised by McDermott and Doyle~\cite{mcdermott-doyle:80} and McDermott~\cite{mcdermott:82}.
  While Moore defined autoepistemic logic without referring to any particular modal system, 
  it turned out that his logic coincides with the non-monotonic modal logic based on KD45~\cite{shvarts:90}.
  Therefore, autoepistemic logic can be considered a popular representative among the 
  non-monotonic modal logics. 
  The connection of these logics and particularly autoepistemic logic to default logic has been extensively studied. 
  The first major approach in this direction was taken by 
  Konolige~\cite{konolige:88}, who showed that default logic can be embedded into autoepistemic logic using 
  slightly different semantics for the latter. Subsequently, Marek and Truszczynski~\cite{marek-truszczynski:89,marek-truszczynski:90} 
  showed that, using strengthened notions in autoepistemic logic or weakened notions in default logic, the two logics coincide in terms of expressivity. Finally, Gottlob~\cite{gottlob:95} showed that default logic can be embedded into standard autoepistemic logic, while the converse direction was shown to hold by Janhunen~\cite{janhunen:99}.

  The intention of Moore was to create a logic modelling the beliefs of an ideally rational agent, \ie, an agent that believes all things he can deduce and refutes belief in everything else. To this end, autoepistemic logic extends propositional logic with the modal operator $L$ stating that its argument ``is believed''. The set of all autoepistemic formulae $\allAutoepistemicFormulae$ is defined as
  \[
    \varphi ::= a \mid f(\varphi,\ldots,\varphi) \mid L \varphi,
  \]
  where $a$ is a proposition and $f$ is a Boolean function, 
  and the relation $\models$ is extended to simply treat formulae starting with an $L$ as atomic. 
  Similarly to default logic, the semantics of autoepistemic logic are defined in terms of fixed points, which in the context of autoepistemic logic are called stable expansions:

  \begin{definition}[\cite{moore:85}]
    Let $\Sigma \subseteq \allAutoepistemicFormulae$ be a set of autoepistemic formulae. 
    A set $\Delta \subseteq \allAutoepistemicFormulae$ is a \emph{stable expansion} of $\Sigma$ if it satisfies the equation
    \[
      \Delta = \theorems{\Sigma \cup L(\Delta) \cup \neg L(\allAutoepistemicFormulae\setminus\Delta)},
    \]
    where $L(\Delta) := \{ L\varphi \mid \varphi \in \Delta\}$ and $\neg L(\allAutoepistemicFormulae\setminus\Delta) := \{ \neg L\varphi \mid \varphi \not\in \Delta\}$.
  \end{definition}

  \begin{example} \label{ex:ael-1}
    Consider the set $\Sigma=\{ Lx \lor y, x \lor Ly, L(x \lor y) \to z\}$ of autoepistemic formulae.
    We claim that $\Sigma$ has two stable expansions, each of which containing $z$.
    Sticking with our informal interpretation of autoepistemic logic as the logic of an ideally rational agent's beliefs, 
    observe that we cannot deduce $x$ from $\Sigma$.
    Hence, we would assign $Lx$ the value $\false$ and consequently be able to derive $y$ from $Lx \lor y$.
    This in turn allows us to conclude $z$ from $L(x \lor y) \to z$, as $x \lor y$ is derivable from $y$.
    Indeed, 
    the set $\bigcup_{i \in \N} \Delta_i$ of formulae recursively defined via $\Delta_0:=\{x\}$ and $\Delta_i := \theorems{\Sigma \cup L(\Delta_{i-1}) \cup \neg L(\allAutoepistemicFormulae\setminus\Delta_{i-1})}$ is a stable expansion of $\Sigma$ that contains $z$.

    On the other hand, we are not able to deduce $y$ from $\Sigma$ either. 
    Hence, we could also continue to assign to $Ly$ the value $\false$ and be therefore able to derive $x$.
    Again, we may conclude $z$ from $L(x \lor y) \to z$.
    And just as above, $\bigcup_{i \in \N} \Delta'_i$ with $\Delta'_0:=\{y\}$ and $\Delta'_i$ defined as $\Delta_i$ is a stable expansion of $\Sigma$ that contains $z$.
  \end{example}

  There is an important difference to default logic as stable expansions need not be minimal fixed points:

  \begin{example}
    Consider $\Sigma':= \{Lp \to p\}$. 
    The set $\Sigma'$ has two stable expansions,
    one stable expansion containing $\neg Lp$ and the other one containing $Lp$.
    As an iterative construction as in Example~\ref{ex:ael-1} is deemed to fail for the latter, 
    it may be considered \emph{ungrounded} in the set of premises $\Sigma$.
  \end{example}

  Clearly, sets of autoepistemic formulae can also posses no or a single stable expansion. Hence,
  the \emph{expansion existence problem}, the \emph{credulous reasoning problem} and the \emph{skeptical reasoning problem} arise just as in default logic.
  The first treatment of the complexity of these problems has been performed by Niemel\"a~\cite{niemela:90}. 
  In his paper, he gave a finite characterization of stable expansions in terms of \emph{full sets}:
  Let $\SFL{\Sigma}$ denote the set of $L$-prefixed subformulae of formulae in $\Sigma$.
  \begin{definition}[\cite{niemela:90}] \label{def:ae-full-set}
    Let $\Sigma \subseteq \allAutoepistemicFormulae$ be a set of autoepistemic formulae.
    A set $\Lambda \subseteq \SFL{\Sigma} \cup \{\neg L\varphi \mid L\varphi \in \SFL{\Sigma}\}$ is \emph{$\Sigma$-full} if for all $L\psi \in \SFL{\Sigma}$:
    \begin{itemize} \itemsep 0pt
      \item $\Sigma \cup \Lambda \models \psi$ iff $L\psi \in \Lambda$.
      \item $\Sigma \cup \Lambda \not\models \psi$ iff $\neg L\psi \in \Lambda$.
    \end{itemize}
  \end{definition}

  \begin{proposition}[\cite{niemela:90}]
    Let $\Sigma \subseteq \allAutoepistemicFormulae$ be a set of autoepistemic formulae.
    If $\Lambda$ is a $\Sigma$-full set, then 
    there exists exactly one stable expansion of $\Delta$ such that $\Lambda \subseteq L(\Delta) \cup \neg L(\allAutoepistemicFormulae\setminus\Delta)$.
    Vice versa, if $\Delta$ is a stable expansion of $\Sigma$, then 
    there exists exactly one $\Sigma$-full set $\Lambda$ such that 
    such that $\Lambda \subseteq L(\Delta) \cup \neg L(\allAutoepistemicFormulae\setminus\Delta)$.
  \end{proposition}

  Using full sets as finite representations for stable expansions, Niemel\"a obtained a $\SigmaP2$ upper bound for the expansion existence problem and the credulous reasoning problem, and a $\PiP2$ upper bound for the skeptical reasoning problem: for the expansion existence problem, simply guess a candidate for a  full set and verify the conditions given in Definition~\ref{def:ae-full-set} using an oracle for formula implication. To extend this idea to the credulous and skeptical reasoning problem, one still needs to define a consequence relation $\models_L$ that, given $\Sigma \subseteq \allAutoepistemicFormulae$ and a $\Sigma$-full set $\Lambda$ implies exactly those formulae contained in the stable expansion corresponding to $\Lambda$. Niemel\"a shows that $\models_L$ can be defined such that the problem of deciding the relation Turing reduces to the implication problem. From this it is easy to see that the credulous and skeptical reasoning problem are contained in $\SigmaP2$ and $\PiP2$ respectively.
  The matching lower bounds were later established by Gottlob in~\cite{gottlob:92}.

  \begin{theorem}[\cite{niemela:90,gottlob:92}] \label{thm:ael-gottlob}
    The expansion existence problem and the credulous reasoning problem for autoepistemic logic are $\SigmaP2$-complete, 
    whereas the skeptical reasoning problem for autoepistemic logic is $\PiP2$-complete.
  \end{theorem}

  The hardness is obtained using a surprisingly simple reduction form the validity problem for quantified Boolean formulae of the form $\exists x_1 \cdots \exists x_n \forall y_1 \cdots \forall y_m \psi$, where $\psi$ is a propositional formula with $\Vars{\psi}=\{x_i \mid 1 \leq i \leq n\} \cup \{y_i \mid 1 \leq i \leq m\}$. 
  Given a formula of the above form, we transform it into a set $\Sigma$ of autoepistemic formulae defined as 
  \[
    \Sigma := \{Lx_i \leftrightarrow x_i \mid 1 \leq i \leq n\} \cup \{L\psi\}.
  \]
  The idea behind the reduction is to mimic existential quantification using different sets of beliefs such that
  an assignment satisfying $\forall y_1 \cdots \forall y_m \psi$ results in a stable expansion:
  If $\sigma$ is an assignment that satisfies $\forall y_1 \cdots \forall y_m \psi$, 
  then the $\Sigma \cup \Lambda$ with $\Lambda = \{L\psi\} \cup \{Lx_i \mid 1 \leq i \leq n\}\cup\{\neg Lx_1 \mid 1 \leq i \leq n\}$ entails
  $\psi$; therefore, $\Lambda$ is a $\Sigma$-full set. On the other hand, if $\Lambda$ is a full set, then we can reconstruct from it an assignment satisfying $\forall y_1 \cdots \forall y_m \psi$.

  Beyond this, the complexity of these problems for fragments of autoepistemic logic has seemingly only been studied 
  in \cite{creignou-meier-thomas-vollmer:10}. There, it was shown that already autoepistemic logic using only $\land$ and $\lor$ is $\SigmaP2$-complete and that tractable fragments occur only for affine sets of Boolean functions.
  To present the results, say that, for a finite set $B$ of Boolean functions, an \emph{autoepistemic $B$-formula} is an autoepistemic formula using Boolean functions from a given finite set $B$ only, and denote by $\allAutoepistemicFormulae(B)$ the set of all autoepistemic $B$-formulae. 
  Further, let \emph{$B$-autoepistemic logic} denote autoepistemic logic restricted to autoepistemic $B$-formulae
  and define the \emph{credulous (resp.\ skeptical) reasoning problem for $B$-autoepistemic logic} as the problem to decide, given a set $\Sigma$ of autoepistemic $B$-formulae and an autoepistemic $B$-formula $\varphi$, whether $\varphi$ is contained in a stable expansion (resp.\ all stable expansions)  of $\Sigma$ .

  \begin{theorem}[\cite{creignou-meier-thomas-vollmer:10}] \label{thm:ael-exp}
    Let $B$ be a finite set of Boolean functions.
    Then the expansion existence problem and the credulous (resp.\ skeptical) reasoning problem for $B$-autoepistemic logic are 
    \begin{enumerate} \itemsep 0pt
      \item 
      $\SigmaP2$-complete (resp.\ $\PiP2$-complete) if $\CloneD_2\subseteq [B]$ or $\CloneS_{00}\subseteq [B]$ or $\CloneS_{10}\subseteq [B]$,
      \item 
      $\NP$-complete (resp.\ $\co\NP$-complete) if $\CloneV_2\subseteq [B]\subseteq \CloneV$,
      \item 
      $\ParityL$-hard and contained in $\P$ if $\CloneL_2\subseteq [B]\subseteq \CloneL$, 
      \item 
      in $\L$ in all other cases (that is, if $[B]\subseteq \CloneE$),
      \end{enumerate}
      via logspace many-one reductions.
  \end{theorem}

  Note that the complexity classification of these problems substantially differs from their analogues in default logic, 
  which can be credited to the different approach to modelling non-monotonicity: while default logic is limited to consistency 
  testing in the justification of a default rule, autoepistemic logic is capable of both positive and negative introspection.
  As another result, in general the intertranslatability of autoepistemic logic and default logic does not hold for fragments of these logics
  (unless collapses considered unlikely occur).

  We will briefly present the ideas behind Theorem~\ref{thm:ael-exp}.
  To start with, the proof relies on the following lemma, which significantly reduces the number of clones to be considered.

  \begin{lemma} \label{lem:ael-constants}
    Let $B$ be a finite set of Boolean functions and $\Sigma \subseteq \allAutoepistemicFormulae(B \cup \{\false,\true\})$.
    Then we can construct in logspace a set $\Sigma' \subseteq \allAutoepistemicFormulae(B)$ such that the stable expansions 
    of $\Sigma$ and $\Sigma'$ coincide on all autoepistemic formulae over $\Vars{\Sigma}$.
  \end{lemma}
  \begin{proof}
    Let $\Sigma \subseteq \allAutoepistemicFormulae(B \cup \{\false,\true\})$ be given. 
    We first eliminate the constant $\true$ using Lemma~\ref{lem:discard_true} and transform the resulting set $\Sigma'$ to 
    $\Sigma''$ by substituting all occurrences of $\false$ with the formula $Lf$, where $f$ is a fresh proposition.
    Suppose that $\Delta$ is a consistent stable expansion of $\Sigma''$. 
    As $f$ cannot be derived from $\Sigma''$, $\Delta$ has to contain the $\neg Lf$. 
    Hence any satisfying assignment of $\Delta$ sets $Lf$ to $\false$.
  \end{proof}

  It thus suffices to consider the complexity of the expansion existence problem for sets $B$ such that $[B] \in \{\CloneBF,\CloneM,\CloneE,\CloneV,\CloneL,\CloneN,\CloneI\}$. The key observation in the proof of Theorem~\ref{thm:ael-exp} is that the reductions from the validity problem for quantified Boolean formulae of the form $\exists x_1\cdots \exists x_n\forall y_1\cdots \forall y_m \psi$ with $\psi$ in negation normal form does not requires negations:
  Given a formula $\varphi$ in the above form, replace all negative literals $\neg x_i$ and $\neg y_i$ in  $\psi$ with new propositions $x_i'$ and $y_i'$. Call the resulting formula $\varphi'$. We then construct the set of autoepistemic $\{\land,\lor\}$-formulae as 
  \[
    \Sigma:=\{\L\varphi'\} \cup \{Lx_i \lor x_i',x_i \lor Lx_i' \mid 1 \leq i \leq n\} \cup \{y_i \lor y_i' \mid 1 \leq i \leq m\}.
  \]
  Due to the formulae $Lx_i \lor x_i',x_i \lor Lx_i'$, $1 \leq i \leq n$, any stable expansion $\Delta$ of $\Sigma$ contains either $x_i$ or $x_i'$ (but not both), while the formulae $y_i \lor y_i'$, $1 \leq i \leq m$, guarantee that either $y_i$ or $y_i'$ is set to $\true$ in any satisfying assignment of $\Delta$. From this, it is easy to see that $\varphi$ is valid iff $\Sigma$ has a stable expansion.
  Hence, the expansion existence problem is $\SigmaP2$-complete for $\CloneS_{00}\subseteq [B]$, $\CloneS_{10}\subseteq [B]$, or $\CloneD_2\subseteq [B]$.

  This construction can be generalized to also work for quantified Boolean formulae of the form $\exists x_1\cdots \exists x_n \psi$ with $\psi$ in conjunctive normal form. Hence, we obtain $\NP$-hardness for $\CloneV_2 \subseteq [B]$. The corresponding upper bound follows from the fact the formula implication of $B$-formulae with $[B] \subseteq \CloneV$ is tractable.

	For an argument to obtain a polynomial-time upper bound for the remaining cases, we refer the reader to the original paper.
	This concludes the discussion of Theorem~\ref{thm:ael-exp}.
  
  Turning to the problem of counting the number of stable expansions, the reader may easily convince himself that the reductions used to establish the $\NP$-hardness and $\SigmaP2$-hardness above are parsimonious. Hence, the complexity classification of the counting problem is analogous to that of the expansion existence problem:
  
  \begin{theorem}[\cite{creignou-meier-thomas-vollmer:10}] \label{thm:num-expansions}
    Let $B$ be a finite set of Boolean functions. 
    Then the problem to count the number of stable expansions in $B$-autoepistemic logic is
    \begin{enumerate} \itemsep 0pt
    \item 
    $\SharpCoNP$-complete if $\CloneD_2\subseteq [B]$ or $\CloneS_{00}\subseteq [B]$ or $\CloneS_{10}\subseteq [B]$,
    \item 
    $\SharpP$-complete if $\CloneV_2\subseteq [B]\subseteq \CloneV$,
    \item 
    in $\FP$ in all other cases (that is, if $[B]\subseteq \CloneL$ or $[B]\subseteq \CloneE$),
    \end{enumerate}
    via parsimonious reductions.
  \end{theorem}

  Note that again parsimonious reductions are sufficient to obtain the completeness results in Theorem~\ref{thm:num-expansions}.

\section{Circumscription}

  The third non-monotonic logic we will turn to is circumscription, which instead of extending classical logic with 
  default rules or introspection restricts the attention to minimal models.
  Circumscription was introduced by McCarthy~\cite{mccarthy:80} in 1980 to overcome the \emph{qualification problem}, \ie, 
  the problem of listing all preconditions required for an action to have its intended effect. 
  His approach was to allow for the conclusion that the objects that can be shown to have a certain property by reasoning 
  are all objects that satisfy this property. Following 
  \cite{lifschitz:85}, this is achieved by considering only those models that are minimal with respect to a preorder on the set of assignments.
  For ease of notation, we will identify assignments $\sigma$ with the set $\{ p \mid \sigma(p)=\true\}$.
    
  \begin{definition}
    Let $P$, $Q$, $Z$ partition the set of propositions and 
    let $\sigma,\sigma' \colon P \cup Q \cup Z \to \{\false,\true\}$ be assignments. 
    Define $\leq_{(P,Z)}$ as the preorder defined by 
    \[
      \sigma \leq_{(P,Z)} \sigma' \iff \sigma \cap P \subseteq \sigma' \cap P  \mbox{ and } \sigma \cap Q = \sigma' \cap Q.
    \]
  \end{definition}

  Using $\leq_{(P,Z)}$, we define a consequence relation $\models_{(P,Z)}$ such that for an assignment $\sigma \colon P \cup Q \cup Z \to \{\false,\true\}$ and a set of formulae $\Gamma$ with $\Vars{\Gamma}\subseteq P\cup Q\cup Z$, $\sigma \models_{(P,Z)} \Gamma$ if
  $\sigma$ is minimal w.r.t.\ $\leq_{(P,Z)}$ among all models of $\Gamma$.. In this case, $\sigma$ is also  called a \emph{circumscriptive model} of $\Gamma$.
  Accordingly, we can define the notion of (circumscriptive) implication:
  $\Gamma \models_{(P,Z)} \varphi$ if $\varphi$ is satisfied in all circumscriptive models of $\Gamma$.
  It is not hard to see that circumscription coincides with reasoning under the \emph{extended closed world assumption}, 
  in which all formulae involving only propositions from $P$ that cannot be derived from $\Gamma$ are assumed to be false \cite{geprpr89}.
  
  \begin{example}
    Let $P:=\{x\}$, $Q:=\emptyset$, $Z:=\{y,z\}$ and $\Gamma:=\{(x \land \neg y) \to z\}$.
    The models of $\Gamma$ are $\emptyset$, $\{y\}$, $\{z\}$, $\{x,y\}$, $\{y,z\}$, $\{x,z\}$, $\{y\}$, and $\{x,y,z\}$.
    Of these, only $\emptyset$, $\{z\}$, and $\{y,z\}$ are minimal with respect to $\leq_{(P,Z)}$.
    Hence, $\Gamma \cup \{y\} \models_{(P,Z)} z$, while $\Gamma \cup \{y\} \not\models z$.
  \end{example}
  
  The notions of circumscriptive models and circumscriptive inference naturally lead to the following decision problems, 
  that received extensive study in the literature:
  
  (\emph{Circumscriptive model checking}) Given a set of formulae $\Gamma$, a preorder $\leq_{(P,Z)}$ on the set of its propositions and an assignment $\sigma$, 
  does $\sigma \models_{(P,Z)} \Gamma$ hold?
  
  (\emph{Circumscriptive inference}) Given a set of formulae $\Gamma$ and a formula $\varphi$, a preorder $\leq_{(P,Z)}$ on the set of their propositions, 
  does $\Gamma \models_{(P,Z)} \varphi$ hold?

  The circumscriptive model checking is dual to a generalization of the minimal satisfiability problem, 
  \ie, the question whether a given formula has a model that is strictly smaller than a given assignment with respect to a given preorder $\leq_{(P,Z)}$, and known to be  $\co\NP$-complete in general~\cite{cadoli:92}, whereas the circumscriptive inference problem 
  was shown to be $\PiP2$-complete by Eiter and Gottlob in~\cite{eiter-gottlob:93}. These results reveal that, alike default logic and autoepistemic logic, circumscription exhibits an increase in the complexity of model checking and reasoning as compared to traditional propositional logic. This increase in the complexity raises the question for restrictions that lower the complexity of these tasks.
  Accordingly, the complexity of these problems has been studied for both restricted sets of Boolean functions and in Schaefers framework.
  We will consider the restrictions obtained from \emph{Schaefer's framework} first.
  
  Define the \emph{circumscriptive model checking problem for sets of relations from $S$} as the problem to decide,
  given a
    a set $\Gamma$ of applications of relations from $S$,
    an assignment $\sigma\colon \Vars{\Gamma} \to \{\false,\true\}$ and 
    a partition $(P,Q,Z)$ of $\Vars{\Gamma}$,
  whether $\sigma$ is a minimal model of $\Gamma$ with respect to $\leq_{(P,Z)}$.
  In~\cite{kiko01a}, Kirousis and Kolaitis showed that using Schaefer's framework, the circumscriptive model checking problem restricted to $Q=Z=\emptyset$ is dichotomic, a result which was later generalized to the general case in~\cite{kiko03}:
  
  \begin{theorem}[\cite{kiko03}]
    Let $S$ be a set of relations.
    Then the circumscriptive model checking problem for sets of relations from $S$ is 
    \begin{enumerate} \itemsep 0pt
      \item $\co\NP$-complete if $S$ is not Schaefer and
      \item in $\P$ in all other cases.
    \end{enumerate}
  \end{theorem}
  
  The tractability if $S$ is Schaefer is easy to verify.
  In this case, the circumscriptive model checking problem Turing reduces to 
  the satisfiability problem, which in this case is tractable by~\cite{sch78}.
  To show the $\co\NP$-hardness in all remaining cases, Kirousis and Kolaitis give an involved  three step reduction from $\mathrm{1}$-$\mathrm{in}$-$\mathrm{3\ SAT}$.
  
  In addition to that, Kirousis and Kolaitis also classified for possible sets of available Boolean functions in an unpublished note.
  Define the \emph{circumscriptive model checking problem for sets of $B$-formulae} as the problem to decide, 
  given
    a set $\Gamma \subseteq \allFormulae(B)$, 
    an assignment $\sigma\colon \Vars{\Gamma} \to \{\false,\true\}$ and 
    a partition $(P,Q,Z)$ of $\Vars{\Gamma}$,
  whether $\sigma$ is a minimal model of $\Gamma$ with respect to $\leq_{(P,Z)}$.
  
  \begin{theorem}[\cite{kiko01b}]
    Let $B$ be a finite set of Boolean functions.
    Then the circumscriptive model checking problem for sets of $B$-formulae is 
    \begin{enumerate} \itemsep 0pt
      \item $\co\NP$-complete if $\CloneS_{02} \subseteq [B]$ or $\CloneS_{12} \subseteq [B]$ or $\CloneD_1 \subseteq [B]$, and
      \item in $\P$ in all other cases (that is, if $[B] \subseteq \CloneM$).
    \end{enumerate}
  \end{theorem}
  
  Key to the classification is that if the set $B$ of all available Boolean functions is monotone,
  then the circumscriptive model checking problem Turing-reduces to the 
  the model checking problem for monotone Boolean formulae.
  Given $\Gamma$ and $\sigma$, denote by $\sigma^i$ the assignment obtained by setting all propositions in $Z$ to $\true$ and 
  the $i$th proposition in $P$, which is set to $\true$ under $\sigma$ to $\false$. 
  Then $\sigma \models_{(P,Z)} \Gamma$ iff for all $\sigma^i$ obtained in this way, $\sigma^i \not\models_{(P,Z)} \Gamma$.
  Thus circumscriptive model checking problem for sets of $B$-formulae is tractable if $[B] \subseteq \CloneM$.
  
  On the other hand, the $\co\NP$-hardness follows from the fact that all remaining sets $B$ satisfy $[B \cup \{\false,\true\}]=\CloneBF$, 
  while we can simulate the Boolean constants:
  If 
  $B\subseteq \CloneR_1$, 
  then appealing to Lemma~\ref{lem:discard_true} suffices. 
  If 
  $B \subseteq \CloneR_1$, 
  we use the mapping
  $\Gamma \mapsto \Gamma':= \Gamma[\false/f,\true/t] \cup \big\{t,f \to \bigwedge \Vars{\Gamma}\big\}$,
  where $\Gamma[\false/f,\true/t]$ denotes the set obtained from $\Gamma$ by replacing all occurrences of $\false$ by the fresh proposition $f$ and all occurrences of $\true$ by the fresh proposition $t$.
  Notice that 
  any model of $\Gamma'$ sets $t$ to $\true$ and that $\Gamma'$ is satisfied by exactly those assignments that satisfy $\Gamma$ and additionally the assignment setting all propositions to $\true$. 
  Thus, if 
    $B\subseteq \CloneR_1$, 
  we may replace any Boolean function in the original $\Gamma$ with an equivalent $(B \cup \{0,1\})$-formula.
  
  We point out that the original proof in \cite{kiko01b} builds on results from~\cite{kiko01a}, the sketch above provides an alternative proof.
  
  As for the inference problem, 
  the computational complexity of circumscriptive inference was first studied in 1990 by Cadoli and Lenzerini~\cite{cale94}, who analyzed the complexity of reasoning under various closed world assumptions. 
  They showed that circumscription for various restrictions on the premises, conclusions and the order still remains intractable and also
  identified tractable fragments. Yet the exact complexity of circumscriptive inference remained open until Eiter and Gottlob~\cite{eiter-gottlob:93} proved its $\PiP2$-completeness.
  This result was further refined both in the framework of 
    restricted sets of available Boolean functions as well as 
    in Schaefer's framework.
  For the latter, Kirousis and Kolaitis~\cite{kiko01c} proved a dichotomy separating the $\PiP2$-complete cases and from those in $\co\NP$, and conjectured that the cases in $\co\NP$ could be refined into $\co\NP$-complete and tractable ones.
  While the $\co\NP$-completeness for dual Horn relations or bijunctive relations was known from~\cite{cale94} and
  Durand and Hermann showed that this also holds for for affine relations~\cite{duhe03},
  Nordh finally affirmatively settled Kirousis and Kolaitis' conjecture in~\cite{nor05}. 
  For the restricted case of \emph{basic circumscription} that requires $Q$ or $Z$ or both to be empty, the trichotomy was established in~\cite{duheno09}.
  We state here the result from~\cite{nor05}.
  
  Say that a relation $R$ is \emph{negative Horn} if it coincides with the set of models of a Horn formula without positive literals.
  Define the \emph{circumscriptive inference problem for sets of relations from $S$} as the problem to decide, 
  given 
    a set $\Gamma$ of applications of relations from $S$,
    a clause $\varphi$ and
    a partition $(P,Q,Z)$ of $\Vars{\Gamma \cup \{\varphi\}}$,
  whether $\Gamma \models_{(P,Z)} \varphi$.
  
  \begin{theorem}[\cite{nor05}]
    Let $S$ be a set of relations.
    Then the circumscriptive inference problem for sets of relations from $S$ is 
    \begin{enumerate} \itemsep 0pt
      \item $\PiP2$-complete if $S$ is not Schaefer,
      \item $\co\NP$-complete if $S$ is Schaefer but neither negative Horn nor both bijunctive and affine nor both Horn and dual Horn, and
      \item in $\P$ in all other case (that is, if $S$ is negative Horn or both bijunctive and affine or both Horn and dual Horn).
    \end{enumerate}
  \end{theorem}
  
  For the approach parameterizing by the set available Boolean functions, 
  the circumscriptive inference problem for sets of $B$-formulae was classified in~\cite{thomas09}.
  Define the \emph{circumscriptive inference problem for sets of $B$-formulae} as the problem to decide,
  given
    a set $\Gamma \subseteq \allFormulae(B)$, 
    a clause $\varphi$ and 
    a partition $(P,Q,Z)$ of $\Vars{\Gamma \cup \{\varphi\}}$,
  whether $\Gamma \models_{(P,Z)} \varphi$.

  \begin{theorem}[\cite{thomas09}]
    Let $B$ be a finite set of Boolean functions.
    Then the circumscriptive inference problem for sets of $B$-formulae is 
    \begin{enumerate} \itemsep 0pt
      \item $\PiP2$-complete if $\CloneS_{02} \subseteq [B]$ or $\CloneS_{12} \subseteq [B]$ or $\CloneD_1 \subseteq [B]$, 
      \item $\co\NP$-complete if $X \subseteq [B] \subseteq Y$ for some
      $X \in \{\CloneV_2,\CloneS_{10},\CloneD_2,\CloneL_2\}$ and $Y  \in \{\CloneM,\CloneL\}$, and
      \item in $\P$ in all other case (that is, if $[B] \subseteq \CloneN$ or $[B] \subseteq \CloneE$).
    \end{enumerate}
  \end{theorem}
  
  \begin{remark}
    Unlike the reasoning problems defined for $B$-default logic and $B$-autoepistemic logic,
    here an arbitrary formula is to be tested for implication. 
    The complexity of deciding the circumscriptive inference of a $B$-formula from a set of $B$-formulae 
    is the same as above for all sets $B$ except those satisfying $\CloneL_2 \subseteq [B] \subseteq \CloneL$; 
    for these the problem is only known to be $\ParityL$-hard and contained in $\co\NP$.
  \end{remark}

  In the remainder of this section, 
  we will consider the counting complexity of circumscription, namely the problem to count the number of 
  minimal models of a given set $\Gamma$ w.r.t.\ to a given preorder $\leq_{(P,Z)}$. This problem, henceforth referred to as the
  \emph{circumscriptive model counting problem}, has recently gained a lot of interest. 
  Its restriction to $Q=Z=\emptyset$ is equivalent to the \emph{minimal model counting problem}, \ie, 
  the problem of counting the number of minimal models w.r.t.\ the coordinatewise partial order 
  on assignments induced by $\false < \true$. While the problem to count the number of all models of a given formula is 
  well-known to be $\SharpP$-complete via parsimonious reduction~\cite{val79b}, 
  the exact complexity of the minimal model counting problem was open for a long time. 
  The problem is easily seen to belong to $\SharpCoNP$, since deciding whether a given assignment 
  is among the minimal models of a given set of formulae is in $\co\NP$.
  Similarly, its $\SharpP$-hardness via parsimonious reductions is apparent: 
  the mapping $\varphi \mapsto \big\{\varphi \land \bigwedge_{x_i \in \Vars{\varphi}} (x_i \xor y_i)\big\}$
  with $y_i \notin \Vars{\varphi}$ constitutes a  parsimonious reduction from the problem to count the number 
  of all models of a given formula. However, $\SharpCoNP$-hardness could only be established using reductions 
  under which not all of the classes of the counting polynomial hierarchy are closed. 
  
  In 2000, Durand, Hermann and Kolaitis~\cite{duheko05} introduced the notion of \emph{subtractive reductions} 
  which suitably relaxed the notion of parsimonious reductions. 
  Say that a counting problem $\#A$ reduces to a counting problem $\#B$ via strong subtractive reduction 
  if there exists a pair of polynomial-time computable functions 
  $f,g$ such that for all $x$,
  $B(g(x)) \subseteq B(f(x))$ and $|A(x)|=|B(f(x))| - |B(g(x))|$.
  Subtractive reductions are the transitive closure of strong subtractive reductions, \ie,
  $\#A$ reduces to $\#B$ via subtractive reductions
  if there exists a finite sequence $(\#C)_{1 \leq i \leq n}$ such that 
  $\#C_1=\#A$, $\#C_n=\#B$ and $\#C_i$ reduces to $\#C_{i+1}$ via strong subtractive reduction for all $1 \leq i < n$.
  Clearly, each parsimonious reduction is also a subtractive reduction.
  And, more importantly, $\SharpP$ and $\#\mathord{\cdot}\PiP{k}$, for all $k \geq 1$, are closed under subtractive reductions.
  
  \begin{theorem}[\cite{duheko05}]
    The minimal model counting problem and the circumscriptive model counting problem are $\SharpCoNP$-complete via subtractive reductions.
  \end{theorem}
  
  The counting complexity of the minimal model counting problem was further studied in~\cite{durand-hermann:08}. There the authors show that, using Schaefer's framework, the restriction to relations that are either dual Horn, bijunctive or affine reduces the complexity of the problem to $\SharpP$-completeness, whereas the restriction to Horn relations or relations that are both bijunctive an affine yields efficiently computable counting problems.
  Hence, the counting problem for the case that all relations are both dual Horn and Horn as well as
  the case that all relations are negative Horn are $\SharpP$-complete while the underlying decision problems are still tractable.
    
  The counting complexity of both the minimal model counting problem and the circumscriptive model counting problem was further studied in~\cite{thomas:10:phd}, where the complexity of the fragments obtained by restricting the set of Boolean functions is classified.
  Here, the counting complexity decreases analogously to the complexity of the underlying decision problem. 
  
  \begin{theorem}[\cite{thomas:10:phd}]
    Let $B$ be a finite set of Boolean functions. 
    Then the minimal model counting problem for $B$-formulae and the circumscriptive model counting problem for $B$-formulae is
    \begin{enumerate} \itemsep 0pt
      \item 
      $\SharpCoNP$-complete via subtractive reductions if $\CloneS_{02} \subseteq [B]$ or $\CloneS_{12} \subseteq [B]$ or $\CloneD_1 \subseteq [B]$,
      
      \item 
      $\SharpP$-complete via subtractive reductions if $\CloneS_{00} \subseteq [B] \subseteq \CloneM$ or $\CloneS_{10} \subseteq [B] \subseteq \CloneM$ or $\CloneD_2 \subseteq [B] \subseteq \CloneM$,
      
      \item 
      $\SharpP$-complete via Turing reductions if $\CloneV_{2} \subseteq [B] \subseteq \CloneV$ or $\CloneL_2 \subseteq [B] \subseteq \CloneL$, and

      \item in $\FP$ in all other cases (that is, if $[B] \subseteq \CloneN$ or $[B] \subseteq \CloneE$).
    \end{enumerate} 
  \end{theorem}

\section{Abduction}

Abduction is a fundamental and important form of non-monotonic reasoning introduced by Peirce~\cite{peirce:55}.
It can be thought of as a form of hypothetical reasoning: 
to ask what can be abduced from an observation $\alpha$ is to ask for an explanation, 
which in conjunction with the given background knowledge accounts for $\alpha$.
The importance of this formalism to artificial intelligence was first emphasized by Morgan~\cite{morgan:71} and 
has been fruitfully used in many areas of computer science such as
medical diagnosis~\cite{byaltajo89}, text analysis~\cite{hostapma93}, system diagnosis~\cite{stwo01},
configuration problems~\cite{amfama02}, temporal knowledge bases~\cite{boli00}, 
and has connections to default reasoning~\cite{selman-levesque:1990}.

Here we will consider logic based abduction in which the background theory is represented by a logical theory, specifically in propositional logic. Hence, the abduction problem can in general be formulated as the problem, given a 
  \emph{knowledge base} $\Gamma \subseteq \allFormulae$, 
  a set $A \subseteq \Vars{\Gamma}$ of propositions called \emph{hypothesis}, and 
  an \emph{observation} $q \in \Vars{\Gamma}$, 
to compute a set $E \subseteq \Lits{A}$ of literals over $A$ such that $\Gamma \cup E$ is consistent and $\Gamma \cup E \models q$.
If such a set $E$ exists, then it is called an \emph{explanation} for the \emph{abduction problem} $\mathcal{P} = (\Gamma,A,q)$. 

The computational complexity of this problem has first been considered by Selman and Levesque~\cite{selman-levesque:1990}, 
who showed that it is $\NP$-hard to compute an explanation if $\Gamma$ is restricted to Horn clauses.
Independently, Friedrich, Gottlob and Nejdl~\cite{fegone90} studied the problem for definite Horn clauses. They show that deciding whether a given proposition is contained in some or all explanations is tractable whereas deciding whether a proposition is contained in a subset-minimal explanation is $\NP$-complete.
It is yet tractable to compute some subset-minimal explanation in this case~\cite{bylander:91}. 
Further results were obtained by Eshghi~\cite{eshghi:93}, who proved that 
finding subset-minimal explanations becomes tractable if $\Gamma$ is acyclic Horn and its pseudo-completion is unit-refutable.
Finally, the complexity of logic-based abduction was settled by Eiter and Gottlob.

\begin{theorem}[\cite{eiter-gottlob:1995}] \label{thm:abd-eiter-gottlob}
  \leavevmode
  \begin{enumerate} \itemsep 0pt
    \item 
    To decide,
    given an abduction problem $\mathcal{P}=(\Gamma,A,q)$ and a set $E \subseteq \Lits{A}$,
    whether $E$ is an explanation for $\mathcal{P}$ is
    $\mathrm{D}^\P$-complete.
    \item 
    To decide,
    given an abduction problem $\mathcal{P}$,
    whether there exists an explanation for $\mathcal{P}$ is $\SigmaP2$-complete.
    \item 
    To decide, 
    given an abduction problem $\mathcal{P}$ and a proposition $p$,
    whether all explanations for $\mathcal{P}$ contain $p$ is $\PiP2$-complete. 
  \end{enumerate}
  These results also hold for subset-minimal explanations.
\end{theorem}

Eiter and Gottlob~\cite{eiter-gottlob:1995} also studied the complexity of abduction for preference relations other than subset-minimality, 
but to include these results here would go beyond the scope of this survey. 
As in most cases the complexity of the first and the third problem in Theorem~\ref{thm:abd-eiter-gottlob} can be derived 
from the question whether an explanation exists, we will henceforth focus on the complexity of deciding the existence of an explanation.

Due to its applications to knowledge-based systems, 
it is natural to consider the complexity of this problem in Schaefer's framework.
Using this approach Creignou and Zanuttini~\cite{crza06} showed that the complexity of 
deciding the existence of an explanation forms a trichotomy.
Define the \emph{explanation existence problem for sets of relations from $S$} as the problem to decide, 
  given a set $\Gamma$ of applications of relations from $S$,
  a set $A \subseteq \Vars{\Gamma}$ of propositions, and 
  a proposition $q \in \Vars{\Gamma} \setminus {A}$, 
whether there exists a set $E \subseteq \Lits{A}$ such that $\Gamma \cup E$ is consistent and $\Gamma \cup E \models q$.
Say that a relation $R$ is \emph{IHS-B$-$} if if it coincides with the models of a CNF formula whose clauses are all of one of the following types: $(x_i)$, $(\neg x_{i_1} \lor x_{i_2})$, $(\neg x_{i_1} \lor \cdots \lor \neg x_{i_k})$ for some $k > 0$.
Analogously say that a relation $R$ is \emph{IHS-B$+$} if if it coincides with the models of a CNF formula whose clauses are all of one of the following types: $(\neg x_i)$, $(\neg x_{i_1} \lor x_{i_2})$, $(x_{i_1} \lor \cdots \lor x_{i_k})$ for some $k > 0$.
Clearly, any IHS-B$-$ formula (resp.\ IHS-B$+$) formula is Horn (resp.\ dual Horn).

\begin{theorem}[\cite{crza06}] \label{thm:abduction-crza}
  Let $S$ be a set of relations.
  Then the explanation existence problem for sets of relations from $S$ is 
  \begin{enumerate} \itemsep 0pt
    \item $\SigmaP2$-complete if $S$ is not Schaefer,
    \item $\NP$-complete if $S$ is either Horn or dual Horn but neither bijunctive nor affine nor definite Horn nor IHS-B$+$ nor IHS-B$-$,
    \item in $\P$ if $S$ is bijunctive or affine or definite Horn or IHS-B$+$ or IHS-B$-$.
  \end{enumerate}
\end{theorem}

As for restricted sets $\Gamma$ the restriction of the query $q$ to a positive proposition does no longer come without loss of generality, 
variants of this problem have been studied, where the observation is a term (a conjunction of literals), a clause or an arbitrary formula.
For example, if one allows for the query to be a literal instead of a proposition, 
then the explanation existence problem for definite Horn relations becomes $\NP$-complete (apart from that, the classification of Theorem~\ref{thm:abduction-crza} remains valid).
For the case that the observation is a term, a trichotomy has been established in~\cite{noza05}:
here only the case that all relations are affine remains tractable, while Horn relations, dual Horn relations and 
relations expressible using either only $(\neg x_{i_1} \lor x_{i_2})$, or only  $(x_{i_1} \leftrightarrow x_{i_2})$ and $(x_{i_1} \lor x_{i_2})$, or only $(x_{i_1} \leftrightarrow x_{i_2})$ and $(\neg x_{i_1} \lor \neg x_{i_2})$ lead to $\NP$-complete fragments.
Further loosening the restriction on the observations to allow for an arbitrary formula leads to a dichotomic classification into $\NP$-complete fragments (if the set of relations is Schaefer) and $\SigmaP2$-complete fragments (in all other cases).
These results are summarized and extended to also cover clauses and arbitrary formulae and to several restrictions on the hypothesis in~\cite{noza08}.

Seeking further insights into the sources of complexity, abduction has also been studied for restricted sets of Boolean functions.
In~\cite{crstth10}, Creignou, Schmidt and Thomas completely classified the complexity of the \emph{explanation existence problem for $B$-formulae}, defined as the problem to decide, given a 
  given a set $\Gamma \subseteq \allFormulae(B)$
  a set $A \subseteq \Vars{\Gamma}$ of propositions, and 
  an observation $q \in \Vars{\Gamma} \setminus {A}$, 
whether there exists a set $E \subseteq \Lits{A}$ such that $\Gamma \cup E$ is consistent and $\Gamma \cup E \models q$.

\begin{theorem}[\cite{crstth10}]
  Let $B$ be a finite set of Boolean functions.
  Then the explanation existence problem for $B$-formulae is 
  \begin{enumerate} \itemsep 0pt
    \item $\SigmaP2$-complete if $\CloneS_{02} \subseteq [B]$ or $\CloneS_{12} \subseteq [B]$ or $\CloneD_1 \subseteq [B]$,
    \item $\NP$-complete if $\CloneS_{00} \subseteq [B]\subseteq \CloneM$ or $\CloneS_{10} \subseteq [B]\subseteq \CloneM$ or $\CloneD_2 \subseteq [B]\subseteq \CloneM$, and
    \item in $\P$ in all other cases (that is, if $[B] \subseteq \CloneE$ or $[B] \subseteq \CloneV$ or $[B] \subseteq \CloneL$).
  \end{enumerate}
\end{theorem}

The complexity of this problem for observations represented by clauses, terms and $B$-formula was also classified.
The relaxation to clauses does not alter the complexity classification; however, the relaxation to observations represented by terms increases the complexity of the cases satisfying $[B \cup \{\false,\true\}] = \CloneV$ to $\NP$-completeness.
Lastly, if the observation is formalized as a $B$-formula then the classification becomes dichotomic with the clones above $\CloneE$, $\CloneV$ or $\CloneL$ being $\SigmaP2$-complete and all remaining clones being tractable; thus skipping the intermediate $\NP$ level.

The complexity of propositional abduction has thus been systematically studied and is well understood. 
The counting complexity of abduction was first studied by Hermann and Pichler~\cite{hermann-pichler:07,hepi08}.
The problems arising in this context are the problem to count the number of all explanations as well as the problem to count the number
of minimal explanations with respect to a given preference relation, \eg, the subset-minimal explanations or explanations of minimal cardinality. 
From the set of counting problems for abduction studied by Hermann and Pichler, we will consider the following two problems:

(\emph{Explanation counting})
Given a set $\Gamma \subseteq \allFormulae$, a set $A \subseteq \Vars{\Gamma}$, and a conjunction $q_1 \land \cdots \land q_n$ or propositions from $\Vars{\Gamma} \setminus {A}$, to count the number of all sets $E \subseteq A$ such that $\Gamma \cup E$ is consistent and $\Gamma \cup E \models q_1 \land \cdots \land q_n$. 

(\emph{Subset-minimal explanation counting})
Given a set $\Gamma \subseteq \allFormulae$, a set $A \subseteq \Vars{\Gamma}$, and a conjunction $q_1 \land \cdots \land q_n$ or propositions from $\Vars{\Gamma}\setminus {A}$, to count the number of all sets $E \subseteq A$ such that $\Gamma \cup E$ is consistent and $\Gamma \cup E \models q_1 \land \cdots \land q_n$.

\begin{theorem}[\cite{hermann-pichler:07}]
  The explanation counting problem and the subset-minimal explanation counting problem are $\SharpCoNP$-complete via subtractive reductions.
\end{theorem}

However, if one restricts $\Gamma$ be be a set of applications of relations from a fixed set $S$ of available relations, then the counting complexity drops by at least one level of the counting polynomial hierarchy:
The explanation counting problem for sets of relations from $S$ is $\SharpP$-complete if $S$ is Horn or dual Horn or bijunctive; it is contained in $\FP$ if $S$ is affine and the explanations are allowed to contain literals instead of propositions.
On the other hand, the subset-minimal explanation counting problem for sets of relations from $S$ is $\SharpP$-complete in all of the previously mentioned cases.

In addition to these results, the complexity of counting all explanations in the case that 
$\Gamma$ is represented by a set of $B$-formulae and
the observation is a single proposition was also studied in~\cite{crstth10}.
There both variants, to count the number of (positive) explanations and to count the number of literal explanations, have been studied.

\begin{theorem}[\cite{crstth10}] \label{thm:abduction-counting-creigou}
  Let $B$ be a finite set of Boolean functions.
  Then the explanation counting problem for sets of $B$-formulae is 
  \begin{enumerate} \itemsep 0pt
    \item $\SharpCoNP$-complete if $\CloneS_{02} \subseteq [B]$ or $\CloneS_{12} \subseteq [B]$ or $\CloneD_1 \subseteq [B]$,
    \item $\SharpP$-complete if $\CloneV_{2} \subseteq [B]\subseteq \CloneM$ or $\CloneS_{10} \subseteq [B]\subseteq \CloneM$ or $\CloneD_2 \subseteq [B]\subseteq \CloneM$, and
    \item in $\FP$ in all other cases (that is, if $[B] \subseteq \CloneE$ or $[B] \subseteq \CloneL$),
  \end{enumerate}
  via subtractive reductions.
  If explanations are restricted to contain positive literals only, then the problem is contained in $\FP$ for $\CloneV_2 \subseteq [B] \subseteq \CloneV$ and only known to be in $\SharpP$ for $\CloneL_2 \subseteq [B] \subseteq \CloneL$. 
\end{theorem}

The open case in Theorem~\ref{thm:abduction-counting-creigou} is equivalent to the case where $\Gamma$ is restricted to a set of affine relations. This case was already left open in~\cite{hermann-pichler:07}.

\newcommand{\etalchar}[1]{$^{#1}$}


\begin{thebibliography}{BMTV09b}

\bibitem[AFM02]{amfama02}
J.~Amilhastre, H.~Fargier, and P.~Marquis.
\newblock Consistency restoration and explanations in dynamic {CSPs}.
\newblock {\em Artif. Intell.}, 135(1-2):199--234, 2002.

\bibitem[BATJ89]{byaltajo89}
T.~Bylander, D.~Allemang, M.~C. Tanner, and J.~R. Josephson.
\newblock Some results concerning the computational complexity of abduction.
\newblock In {\em Proc.\ 1st KR}, pages 44--54, 1989.

\bibitem[BBC{\etalchar{+}}10]{babocrrescvo10}
M.~Bauland, E.~B{\"o}hler, N.~Creignou, S.~Reith, H.~Schnoor, and H.~Vollmer.
\newblock The complexity of problems for quantified constraints.
\newblock {\em Theory Comput. Syst.}, 47(2):454--490, 2010.

\bibitem[BEZ02]{Ben-Eliyahu-Zohary02}
R.~Ben-Eliyahu-Zohary.
\newblock Yet some more complexity results for default logic.
\newblock {\em Artif. Intell.}, 139(1):1--20, 2002.

\bibitem[BL00]{boli00}
M.~Bouzid and A.~Ligeza.
\newblock Temporal causal abduction.
\newblock {\em Constraints}, 5(3):303--319, 2000.

\bibitem[BMTV09a]{beyersdorff-meier-thomas-vollmer:09a}
O.~Beyersdorff, A.~Meier, M.~Thomas, and H.~Vollmer.
\newblock The complexity of propositional implication.
\newblock {\em Inf. Process. Lett.}, 109(18):1071--1077, 2009.

\bibitem[BMTV09b]{bemethvo08}
O.~Beyersdorff, A.~Meier, M.~Thomas, and H.~Vollmer.
\newblock The complexity of reasoning for fragments of default logic.
\newblock In {\em Proc.\ 12th SAT}, volume 5584 of {\em LNCS}, pages 51--64,
  Berlin Heidelberg New York, 2009. Springer.

\bibitem[Byl91]{bylander:91}
T.~Bylander.
\newblock The monotonic abduction problem: A functional characterization on the
  edge of tractability.
\newblock In {\em Proc.\ 2nd KR}, pages 70--77, 1991.

\bibitem[Cad92]{cadoli:92}
M.~Cadoli.
\newblock The complexity of model checking for circumscriptive formulae.
\newblock {\em Inf. Process. Lett.}, 44:113--118, 1992.

\bibitem[CHS07]{chhesc07}
P.~Chapdelaine, M.~Hermann, and I.~Schnoor.
\newblock Complexity of default logic on generalized conjunctive queries.
\newblock In {\em Proc.\ 9th LPNMR}, volume 4483 of {\em Lecture Notes in
  Computer Science}, pages 58--70. Springer, 2007.

\bibitem[CL94]{cale94}
M.~Cadoli and M.~Lenzerini.
\newblock The complexity of propositional closed world reasoning and
  circumscription.
\newblock {\em J. Comput. Syst. Sci.}, 48(2):255--310, 1994.

\bibitem[CMTV10]{creignou-meier-thomas-vollmer:10}
N.~Creignou, A.~Meier, M.~Thomas, and H.~Vollmer.
\newblock The complexity of reasoning for fragments of autoepistemic logic.
\newblock In {\em Circuits, Logic, and Games}, volume 10061 of {\em Dagstuhl
  Seminar Proceedings}, 2010.

\bibitem[CST10]{crstth10}
N.~Creignou, J.~Schmidt, and M.~Thomas.
\newblock Complexity of propositional abduction for restricted sets of
  {B}oolean functions.
\newblock In {\em Proc.\ 12th KR}. AAAI, 2010.

\bibitem[CV08]{crvo08}
Nadia Creignou and Heribert Vollmer.
\newblock Boolean constraint satisfaction problems: When does post's lattice
  help?
\newblock In N.~Creignou, Ph.~G. Kolaitis, and H.~Vollmer, editors, {\em
  Complexity of Constraints}, volume 5250 of {\em Lecture Notes in Computer
  Science}, pages 3--37, Berlin Heidelberg New York, 2008. Springer.

\bibitem[CZ06]{crza06}
N.~Creignou and B.~Zanuttini.
\newblock A complete classification of the complexity of propositional
  abduction.
\newblock {\em SIAM J. Comput.}, 36(1):207--229, 2006.

\bibitem[DH03]{duhe03}
A.~Durand and M.~Hermann.
\newblock The inference problem for propositional circumscription of affine
  formulas is {coNP}-complete.
\newblock In {\em Proc.\ 20th STACS}, volume 2607 of {\em LNCS}, pages
  451--462. Springer, 2003.

\bibitem[DH08]{durand-hermann:08}
A.~Durand and M.~Hermann.
\newblock On the counting complexity of propositional circumscription.
\newblock {\em Inf. Process. Lett.}, 106(4):164--170, 2008.

\bibitem[DHK05]{duheko05}
A.~Durand, M.~Hermann, and P.~G. Kolaitis.
\newblock Subtractive reductions and complete problems for counting complexity
  classes.
\newblock {\em Theoretical Computer Science}, 340(3):496--513, 2005.

\bibitem[DHN09]{duheno09}
A.~Durand, M.~Hermann, and G.~Nordh.
\newblock Trichotomy in the complexity of minimal inference.
\newblock In {\em Proc.\ 24th LICS}, pages 387--396. IEEE Computer Society,
  2009.

\bibitem[EG93]{eiter-gottlob:93}
T.~Eiter and G.~Gottlob.
\newblock Propositional circumscription and extended closed world reasoning are
  {$\Pi_2^p$}-complete.
\newblock {\em Theoretical Computer Science}, 114(2):231--245, 1993.
\newblock Addendum in TCS 118 (1993), page 15.

\bibitem[EG95]{eiter-gottlob:1995}
T.~Eiter and G.~Gottlob.
\newblock The complexity of logic-based abduction.
\newblock {\em J.\ ACM}, 42(1):3--42, 1995.

\bibitem[Esh93]{eshghi:93}
K.~Eshghi.
\newblock A tractable class of abduction problems.
\newblock In {\em Proc.\ 13th IJCAI}, pages 3--8. Morgan Kaufmann, 1993.

\bibitem[FGN90]{fegone90}
G.~Friedrich, G.~Gottlob, and W.~Nejdl.
\newblock Hypothesis classification, abductive diagnosis and therapy.
\newblock In {\em Expert Systems in Engineering}, volume 462 of {\em LNCS},
  pages 69--78. Springer, 1990.

\bibitem[Got92]{gottlob:92}
G.~Gottlob.
\newblock Complexity results for nonmonotonic logics.
\newblock {\em J. Log. Comput.}, 2(3):397--425, 1992.

\bibitem[Got95a]{got95b}
G.~Gottlob.
\newblock {NP} trees and {Carnap's} modal logic.
\newblock {\em J. ACM}, 42:421--457, 1995.

\bibitem[Got95b]{gottlob:95}
G.~Gottlob.
\newblock Translating default logic into standard autoepistemic logic.
\newblock {\em J. ACM}, 42(4):711--740, 1995.

\bibitem[GPP89]{geprpr89}
M.~Gelfond, H.~Przymusinska, and T.~C. Przymusinski.
\newblock On the relationship between circumscription and negation as failure.
\newblock {\em Artificial Intelligence}, 38(1):75--94, 1989.

\bibitem[HP07]{hermann-pichler:07}
M.~Hermann and R.~Pichler.
\newblock Counting complexity of propositional abduction.
\newblock In {\em Proc.\ 20th IJCAI}, volume 4451 of {\em LNCS}, pages
  417--422. Springer, 2007.

\bibitem[HP08]{hepi08}
M.~Hermann and R.~Pichler.
\newblock Counting complexity of minimal cardinality and minimal weight
  abduction.
\newblock In {\em Proc.\ 11th JELIA}, volume 5293 of {\em LNCS}, pages
  206--218. Springer, 2008.

\bibitem[HSAM93]{hostapma93}
J.~R. Hobbs, M.~E. Stickel, D.~E. Appelt, and P.~A. Martin.
\newblock Interpretation as abduction.
\newblock {\em Artif. Intell.}, 63(1-2):69--142, 1993.

\bibitem[HV95]{hevo95}
L.~Hemaspaandra and H.~Vollmer.
\newblock The satanic notations: counting classes beyond {\#}{P} and other
  definitional adventures.
\newblock {\em Complexity Theory Column 8, ACM-SIGACT News}, 26(1):2--13, 1995.

\bibitem[Jan99]{janhunen:99}
T.~Janhunen.
\newblock On the intertranslatability of non-monotonic logics.
\newblock {\em Ann. Math. Artif. Intell.}, 27(1-4):79--128, 1999.

\bibitem[KK01a]{kiko01b}
L.~M. Kirousis and P.~Kolaitis.
\newblock The complexity of minimal satisfiability in {P}ost's lattice.
\newblock Unpublished notes, 2001.

\bibitem[KK01b]{kiko01a}
L.~M. Kirousis and P.~Kolaitis.
\newblock The complexity of minimal satisfiability problems.
\newblock In {\em Proc.\ 18th STACS}, volume 2010 of {\em LNCS}, pages
  407--418. Springer, 2001.

\bibitem[KK01c]{kiko01c}
L.~M. Kirousis and P.~G. Kolaitis.
\newblock A dichotomy in the complexity of propositional circumscription.
\newblock In {\em Proc.\ 16th LICS}, pages 71--80, 2001.

\bibitem[KK03]{kiko03}
L.~M. Kirousis and P.~Kolaitis.
\newblock The complexity of minimal satisfiability problems.
\newblock {\em Information and Computation}, 187(1):20--39, 2003.

\bibitem[Kon88]{konolige:88}
K.~Konolige.
\newblock On the relation between default and autoepistemic logic.
\newblock {\em Artif. Intell.}, 35(3):343--382, 1988.
\newblock Erratum: {\em Artif. Intell.}, 41(1):115.

\bibitem[KS91]{kautz-selman:91}
H.~A. Kautz and B.~Selman.
\newblock Hard problems for simple default logics.
\newblock {\em Artif. Intell.}, 49:243--279, 1991.

\bibitem[Lif85]{lifschitz:85}
V.~Lifschitz.
\newblock Computing circumscription.
\newblock In {\em Proc.\ 9th IJCAI}, pages 121--127. Morgan Kaufman, 1985.

\bibitem[LMS01]{laroussinie-markey-schnoebelen:01}
F.~Laroussinie, N.~Markey, and P.~Schnoebelen.
\newblock Model checking {CTL\textsuperscript{+}} and {FCTL} is hard.
\newblock In {\em Proc.\ 4th FoSSaCS}, volume 2030 of {\em LNCS}, pages
  318--331. Springer, 2001.

\bibitem[LS05]{LiberatoreS05}
P.~Liberatore and M.~Schaerf.
\newblock The complexity of model checking for propositional default logics.
\newblock {\em Data Knowl. Eng.}, 55(2):189--202, 2005.

\bibitem[McC80]{mccarthy:80}
J.~McCarthy.
\newblock Circumscription -- {A} form of non-monotonic reasoning.
\newblock {\em Artif. Intell.}, 13:27--39, 1980.

\bibitem[McD82]{mcdermott:82}
D.~McDermott.
\newblock Non-monotonic logic {II}: nonmonotonic modal theories.
\newblock {\em J. ACM}, 29:33--57, 1982.

\bibitem[MD80]{mcdermott-doyle:80}
D.~McDermott and J.~Doyle.
\newblock Non-monotonic logic {I}.
\newblock {\em Artif. Intell.}, 13:41--72, 1980.

\bibitem[Moo85]{moore:85}
R.~C. Moore.
\newblock Semantical considerations on modal logic.
\newblock {\em Artificial Intelligence}, 25:75--94, 1985.

\bibitem[Mor71]{morgan:71}
C.~G. Morgan.
\newblock Hypothesis generation by machine.
\newblock {\em Artif. Intell.}, 2(2):179--187, 1971.

\bibitem[MT89]{marek-truszczynski:89}
V.~W. Marek and M.~Truszczy{\'n}ski.
\newblock Relating autoepistemic and default logics.
\newblock In {\em Proc.\ 1st KR}, pages 276--288. Morgan Kaufmann, 1989.

\bibitem[MT90]{marek-truszczynski:90}
V.~W. Marek and M.~Truszczy{\'n}ski.
\newblock Modal logic for default reasoning.
\newblock {\em Ann. Math. Artif. Intell.}, 1:175--302, 1990.

\bibitem[Nie90]{niemela:90}
I.~Niemel{\"a}.
\newblock Towards automatic autoepistemic reasoning.
\newblock In {\em Proc.\ 2th JELIA}, volume 478 of {\em LNCS}, pages 428--443.
  Springer, 1990.

\bibitem[Nor05]{nor05}
G.~Nordh.
\newblock A trichotomy in the complexity of propositional circumscription.
\newblock In {\em Proc.\ 11th LPAR}, volume 3452 of {\em LNCS}, pages 257--269.
  Springer, 2005.

\bibitem[NZ05]{noza05}
G.~Nordh and B.~Zanuttini.
\newblock Propositional abduction is almost always hard.
\newblock In {\em Proc.\ 19th IJCAI}, pages 534--539. Professional Book Center,
  2005.

\bibitem[NZ08]{noza08}
G.~Nordh and B.~Zanuttini.
\newblock What makes propositional abduction tractable.
\newblock {\em Artif. Intell.}, 172(10):1245--1284, 2008.

\bibitem[Pei55]{peirce:55}
C.~S. Peirce.
\newblock Abduction and induction.
\newblock In {\em Philosophical writings of Peirce}, chapter~11, pages
  150--156. Dover, New York, 1955.

\bibitem[Pip79]{pip79}
N.~Pippenger.
\newblock On simultaneous resource bounds.
\newblock In {\em Proc.\ 20th FOCS}, pages 307--311. IEEE Computer Society
  Press, 1979.

\bibitem[Pos41]{pos41}
E.~Post.
\newblock The two-valued iterative systems of mathematical logic.
\newblock {\em Ann. Math. Stud.}, 5:1--122, 1941.

\bibitem[Rei80]{reiter:80}
R.~Reiter.
\newblock A logic for default reasoning.
\newblock {\em Artif. Intell.}, 13:81--132, 1980.

\bibitem[Sch78]{sch78}
T.~J. Schaefer.
\newblock The complexity of satisfiability problems.
\newblock In {\em Proc.\ 10th STOC}, pages 216--226. ACM Press, 1978.

\bibitem[Sch07]{schnoor07}
I.~Schnoor.
\newblock {\em The Weak Base Mthod for Constraint Satisfaction}.
\newblock PhD thesis, Leibniz Universit{\"a}t Hannover, 2007.

\bibitem[Shv90]{shvarts:90}
Grigori Shvarts.
\newblock Autoepistemic modal logics.
\newblock In {\em Proc.\ 3rd TARK}, pages 97--109. Morgan Kaufmann, 1990.

\bibitem[SL90]{selman-levesque:1990}
B.~Selman and H.~Levesque.
\newblock Abductive and default reasoning: A computational core.
\newblock In {\em Proc.\ 8th AAAI}, pages 343--348. AAAI Press, 1990.

\bibitem[Sti90]{stillman:90}
J.~Stillman.
\newblock It's not my default: The complexity of membership problems in
  restricted propositional default logics.
\newblock In {\em Proc.\ 8th AAAI}, pages 571--578, 1990.

\bibitem[Sti92]{stillman:92}
J.~Stillman.
\newblock The complexity of propositional default logics.
\newblock In {\em Proc.\ 10th AAAI}, pages 794--800, 1992.

\bibitem[SW01]{stwo01}
M.~Stumptner and F.~Wotawa.
\newblock Diagnosing tree-structured systems.
\newblock {\em Artif. Intell.}, 127(1):1--29, 2001.

\bibitem[Tho09]{thomas09}
M.~Thomas.
\newblock The complexity of circumscriptive inference in {P}ost's lattice.
\newblock In {\em Proc. 10th LPNMR}, volume 5753 of {\em LNCS}, pages 290--302.
  Springer, 2009.

\bibitem[Tho10a]{tho10}
M.~Thomas.
\newblock On the applicability of {P}ost's lattice.
\newblock Technical report, Leibniz Universit{\"a}t Hannover, 2010.
\newblock Available at \href{http://arxiv.org/abs/1007.2924}{arXiv:1007.2924
  [cs.CC]}.

\bibitem[Tho10b]{thomas:10:phd}
M.~Thomas.
\newblock {\em On the complexity of fragments of nonmonotonic logics}.
\newblock PhD thesis, Leibniz Universit{\"a}t Hannover, 2010.

\bibitem[Tod91]{tod91b}
S.~Toda.
\newblock {\em Computational Complexity of Counting Complexity Classes}.
\newblock PhD thesis, Tokyo Institute of Technology, Department of Computer
  Science, Tokyo, 1991.

\bibitem[Val79]{val79b}
L.~G. Valiant.
\newblock The complexity of computing the permanent.
\newblock {\em Theoretical Computer Science}, 8:189--201, 1979.

\end{thebibliography}
\end{document}